\documentclass[12pt]{iopart}

\usepackage{iopams}  
\usepackage{amssymb}
\usepackage{amsthm}
\usepackage{braket}
\usepackage{bm}
\usepackage{graphicx}
\usepackage[hidelinks]{hyperref}
\usepackage{setstack}
\usepackage{iopams}
\usepackage{verbatim}
\usepackage{todonotes}
\usepackage[numbers,sort&compress]{natbib}
\newtheorem{definition}{Definition}[section]
\newtheorem{theorem}{Theorem}[section]
\newtheorem{lemma}{Lemma}[section]
\theoremstyle{remark}
\newtheorem{example}{Example}[section]

\newtheorem{remark}{Remark}[section]
\newcommand{\R}{\mathbb{R}}
\newcommand{\abs}[1]{\left|#1\right|}
\newcommand{\norm}[1]{\left\|#1\right\|}

\newcommand{\eqref}[1]{\eref{#1}}
\newcommand{\dd}{\mathrm{d}}
\newcommand{\Range}{\mathop{\mathrm{Range}}}

\begin{document}

\title[Wandering range of robust quantum symmetries]{Wandering range of robust quantum symmetries}

\author{Daniel Burgarth$^1$, Paolo Facchi$^{2,3}$, Marilena Ligab\`{o}$^4$, Vito Viesti$^{2,3}$ and Kazuya Yuasa$^5$}

\address{$^1$Physics Department, Friedrich-Alexander Universit\"at of Erlangen-Nuremberg,
91058 Erlangen, Germany\\
$^2$Dipartimento di Fisica, Universit\`{a} di Bari, 70126 Bari, Italy\\
$^3$INFN, Sezione di Bari, 70126 Bari, Italy\\
$^4$Dipartimento di Matematica, Universit\`{a} di Bari, 70125 Bari, Italy\\
$^5$Department of Physics, Waseda University, Tokyo 169-8555, Japan}
\ead{vito.viesti@uniba.it}
\vspace{10pt}

\begin{abstract}
   This paper introduces the concept of the wandering range of a robust symmetry $S$ of a Hamiltonian $H$. 
This quantity measures how the perturbed time evolution
$\rme^{\rmi t(H+\varepsilon V)} S \rme^{-\rmi t(H+\varepsilon V)}$
deviates from its unperturbed counterpart
$\rme^{\rmi tH} S \rme^{-\rmi tH} = S$.
Although the wandering range does not necessarily scale linearly with the perturbation strength $\varepsilon$, we identify conditions under which this linear behavior is recovered and we obtain explicit nonperturbative bounds.
\end{abstract}

\section*{Introduction}
In the description of a physical system, symmetries play a crucial role~\cite{symmetry,Wigner_symmetry}. Even though it is not always possible to compute the dynamics exactly, they provide a great deal of information about the evolution. In quantum systems, symmetries are identified with operators that commute with the Hamiltonian of the system~\cite{Sakurai}. Therefore, at least in principle, the knowledge of the Hamiltonian allows one to determine all the symmetries of a quantum system. However, in real physical applications, the Hamiltonian is known only approximately. It is then important to understand whether a symmetry remains stable under perturbations of the Hamiltonian~\cite{kam,Szabo2025robust}.
This stability is also of great practical importance, for instance in analog quantum simulation, where the goal is to reproduce the dynamics of a target system using another controllable quantum platform~\cite{simul2}. Since any realistic implementation is inevitably affected by imperfections and errors, it becomes essential to understand how such inaccuracies influence the predictions of the model. In this context, perturbations of the Hamiltonian provide a natural framework to describe the effect of experimental imperfections~\cite{simul,qsimul}.

It is possible to show that, if we consider the evolution for small times, the symmetries stay close to their initial value. It is a simple application of perturbation theory. However, for long times not all symmetries have the same behavior. There exist symmetries that, despite the perturbation, remain close to their initial value throughout the entire evolution; these are called \emph{robust} symmetries~\cite{kam}. On the other hand, it may happen that the effect of the perturbation accumulates over time and the symmetry acquires a non-negligible deviation from its initial value; these are referred to as \emph{fragile} symmetries. In Ref.~\cite{rob}, we provide a precise algebraic characterization of these two behaviors for quantum systems described by Hamiltonians with pure point spectrum. Furthermore, we introduce the concept of the \emph{wandering range} of a robust symmetry, which quantifies the deviation of the perturbed symmetry from its unperturbed evolution.

The goal of this paper is to analyze the dependence of the wandering range on the strength of the perturbation~$\varepsilon$. In Ref.~\cite{rob}, we have shown that, in general, the wandering range is not of order~$O(\varepsilon)$. The present paper is devoted to identifying conditions under which this linear scaling does hold. We prove that the wandering range, when evaluated on 
states belonging to the dense subspace given by the linear span of the eigenvectors of the unperturbed Hamiltonian, or for a finite-rank symmetry, is indeed of order~$\varepsilon$. 
This is established in Theorem~\ref{th:linear} for admissible perturbations of the Hamiltonian satisfying a regularity condition, and in Theorem~\ref{thm:compactres} for arbitrary linear perturbations of a compact-resolvent Hamiltonian. 

In the second part of the paper, we carry out a detailed analysis of the wandering range of \emph{completely robust symmetries}, that are robust against any bounded perturbations. In this case, we not only prove that the wandering range is uniformly of order~$\varepsilon$, but also derive an explicit bound in terms of the norm of the perturbation  and the minimal spectral gap of the Hamiltonian;
see Theorem~\ref{th:bicomm}\@. The core of the proof is the construction of a new perturbation of the Hamiltonian that commutes with the Hamiltonian itself and generates a dynamics which approximates the perturbed evolution uniformly in time ---an \emph{eternal block-diagonal approximation}~\cite{Eternal}. This result, which is of independent interest, is established in Theorem~\ref{th:eternal}\@. 

The technical framework underlying the argument is based on Theorem~\ref{thm:fund}, whose detailed proof and analysis are given in Section~\ref{convergence}\@. The proof relies on the construction of a unitary block-diagonalizing transformation, known as the Schrieffer-Wolff transformation~\cite{schrieffer_wolff,bravyi_schrieffer_wolff,Cederbaum_1989}, and equivalent in this context to a quantum Kolmogorov-Arnold-Moser (KAM) iteration scheme~\cite{harm,gapless}.

Throughout this analysis, we assume that the Hamiltonian~$H$ is self-adjoint, has pure point spectrum, and possesses a nonvanishing minimal spectral gap. The procedure is based on the \emph{homological equation}~\cite{arnold1989classical}, also known as the commutator equation in the context of matrix analysis~\cite{bhatia1997matrix}. Extending the solution of this equation to unbounded Hamiltonians is a crucial step that allows the iteration to be carried out (Theorem~\ref{l:61}).

The paper is organized as follows.
In Section~\ref{sec:intro}, we briefly recall the description of the quantum evolution in the Heisenberg picture, the definition of conserved quantities, and their strict connection with continuous symmetries. Furthermore, we introduce the concept of wandering range of robust symmetries and recall some results of Refs.~\cite{kam} and~\cite{rob}. In Section~\ref{sec:eigenvector}, we analyze the wandering range evaluated on linear combinations of eigenvectors of the Hamiltonian and/or for finite-rank symmetries. Section~\ref{sec:bounded} deals with the wandering range of completely robust symmetries against uniformly bounded perturbations.  
Section~\ref{convergence} is devoted to the proof of Theorem~\ref{thm:fund} via a quantum KAM iteration scheme.
Firstly we provide a  construction in terms of formal power series,
and then prove their convergence by exploiting the classical sequence of Catalan numbers~\cite{stanley2015catalan}. The section concludes with Example~\ref{ex:josephson}, which illustrates an application of our results to the Hamiltonian of a Josephson junction in an inductive loop.

\section{Notations and preliminaries}\label{sec:intro}
Let $\mathcal{H}$ be a complex separable Hilbert space. We consider closed quantum systems whose time evolution is described by a one-parameter strongly continuous unitary group $\{\rme^{-\rmi tH}\}_{t\in\R}$, where the Hamiltonian operator $H$ is the generator of the group and is a (possibly unbounded) self-adjoint operator with domain $D(H) \subset \mathcal{H}$. In the Heisenberg picture, the evolution of a bounded operator $A \in  B(\mathcal{H})$ is ruled by the following law,
\begin{equation}\label{eqn:HEv}
   t \in \R \mapsto \rme^{\rmi tH}A\rme^{-\rmi tH}.
\end{equation}

There are observables that play a privileged role in the description of a system: conserved
quantities. They are the operators that, despite the evolution described in~\eqref{eqn:HEv}, do not depend on time. Because of Noether's theorem, there is a one to one correspondence between conserved quantities and continuous symmetries of a system: we are going to use the terms `conserved quantities' and `symmetries' interchangeably.
We say that $S \in  B(\mathcal{H})$ is a symmetry of the Hamiltonian $H$ if for all $t \in \R$
\begin{equation}
\rme^{\rmi tH}S\rme^{-\rmi tH}=S.
\end{equation}

The set of all symmetries of $H$ coincides with its commutant,
\begin{equation}
    \{H\}'=\{S \in\mathcal{B(H)}:[S,P_H(\Omega)]=0,\ \mathrm{for\ all}\  \Omega\subset\R\ \mathrm{measurable}\},
\end{equation}
where $\{P_H(\Omega)\}$ is the unique projection-valued measure canonically associated with $H$ by the spectral theorem and $[S,P_H(\Omega)]=SP_H(\Omega)-P_H(\Omega)S$ is the commutator~\cite{teschl2014quantum}. When $H$ is bounded, then $\{H\}'$ is the set of operators commuting with $H$~\cite{BR,ViestiPhD}.

In Refs.~\cite{kam,rob}, a refined classification of symmetries has been introduced in terms of the long-time stability of symmetries with respect to continuous perturbation of the Hamiltonian. More precisely, given a symmetry $S$ of the Hamiltonian $H$, by definition it does not evolve in time if the time evolution is precisely the one generated by $H$. However, different behaviors can be observed if the same symmetry $S$ evolves in time with a Hamiltonian $H(\varepsilon)$ that is a perturbation of $H$ depending continuously on a parameter $\varepsilon$, which we assume small.

\begin{definition}[Perturbation]
\label{def:perturbation}
A family $H(\varepsilon)$ of  operators on $\mathcal{H}$ is a \emph{perturbation of the Hamiltonian $H$} if
\begin{itemize}
\item $H(0)=H$;
\item $H(\varepsilon)= H(\varepsilon)^\dagger$ is self-adjoint on $D(H(\varepsilon)) = D(H)$ for all $\varepsilon\in[0,1]$;
\item the map $\varepsilon \mapsto H(\varepsilon)$ is strongly continuous at $\varepsilon = 0$, namely, for all $\psi\in D(H)$,
\begin{equation}
    \|H(\varepsilon)\psi - H\psi\| \to 0, \quad \mathrm{as} \quad \varepsilon \to 0.
    \label{eq:strongcont}
\end{equation}
\end{itemize}
\end{definition}
The simplest example of perturbation of $H$ is the following one.
\begin{example}[Linear perturbation]\label{exampleLD}
Let $V$ be a $H$-bounded operator, namely $D(H) \subset D(V)$ and there are nonnegative constants $a, b \geq 0$ such that, for all $\psi \in D(H)$, $\|V\psi\| \leq a \|H \psi \|+ b \| \psi\|$. We assume for the sake of simplicity that $a <1$. If $V$ is Hermitian, then, by the Kato-Rellich Theorem~\cite{Kato1995}, for all $\varepsilon\in [0,1]$, the operator $H(\varepsilon)=H+\varepsilon{V}$ is self-adjoint, hence $H(\varepsilon)$ is a perturbation of $H$, that we call \emph{linear perturbation of $H$}.
\end{example}
\begin{remark}
    By mimicking the linear perturbation case, in the following it will be sometimes convenient to write a general perturbation in the form
    \begin{equation}
        H(\varepsilon) = H + \varepsilon V(\varepsilon),
    \end{equation}
and the properties of the family $H(\varepsilon)$ in Definition~\ref{def:perturbation} imply that
    \begin{itemize}
\item $V(\varepsilon)$ is Hermitian on $D(V(\varepsilon)) = D(H)$ for all $\varepsilon\in[0,1]$;
\item for all $\psi\in D(H)$,
$\|V(\varepsilon)\psi\| = o(1/\varepsilon)$,  as $\varepsilon \to 0$.
\end{itemize}
\end{remark}

When the Hamiltonian undergoes a perturbation $H(\varepsilon)$, the symmetries are no longer conserved with respect to the perturbed dynamics, namely,
\begin{equation}
    \rme^{\rmi tH(\varepsilon)}S\rme^{-\rmi tH(\varepsilon)}\neq S,\quad\mathrm{for}\quad S\in\{H\}'.
\end{equation}
We now introduce a quantity that measures the maximal deviation of $S$ from its initial value along the perturbed evolution.

\begin{definition}[Wandering range]
Let $S \in \{H\}'$ be a symmetry of the Hamiltonian $H$, and let $H(\varepsilon)$, with $\varepsilon\in[0,1]$, be a perturbation of $H$. 
For all $\psi \in \mathcal{H}$ and all $\varepsilon \in [0,1]$, we define the \emph{wandering range of $S$ under $H(\varepsilon)$} on $\psi$ as
\begin{equation}
   \delta_{H(\varepsilon)}(S,\psi)=\sup_{t\in\R}
   \|(\rme^{\rmi tH(\varepsilon)}S\rme^{-\rmi tH(\varepsilon)}-S)\psi\|.
\end{equation}
\end{definition}

The behavior of symmetries under perturbations is not uniform. Some symmetries accumulate deviations over long times and eventually differ significantly from their initial values; these are called \emph{fragile symmetries}. On the other hand, there are symmetries that remain close to their initial values throughout the evolution despite the perturbation; these are called \emph{robust symmetries}. Mathematically, this is equivalent to say that the wandering range of a robust symmetry goes to zero as $\varepsilon\to0$.

We now formalize this notion.
\begin{definition}\label{def:Pfr}
Let $S \in \{H\}'$ be a symmetry of the Hamiltonian $H$, and let $H(\varepsilon)$ be a perturbation of $H$. We say that $S$ is \emph{robust against $H(\varepsilon)$} if for all $\psi \in \mathcal{H}$
\begin{equation}\label{eq:rob}
\delta_{H(\varepsilon)}(S,\psi) \to 0, \quad \text{as} \quad \varepsilon \to 0. 
\end{equation} 
\end{definition}
\begin{remark}
If a symmetry $S \in \{H\}'$ is robust against a linear perturbation $H(\varepsilon)=H+\varepsilon V$, we say that $S$ is \emph{$V$-robust}.
\end{remark}
The goal of the paper is to discuss the speed of convergence of the limit in~\eqref{eq:rob}.
To achieve this, we aim to derive an explicit upper bound for the wandering range of the perturbed evolution $\rme^{\rmi tH(\varepsilon)}S\rme^{-\rmi tH(\varepsilon)}$ of the symmetry $S$ around its unperturbed evolution $S$, on the state $\psi$.

If the Hilbert space $\mathcal{H}$ is a finite-dimensional Hilbert space and the perturbation is linear, $H(\varepsilon)=H+\varepsilon{V}$, there is an explicit (uniform) upper bound 
\begin{equation}\label{eq:bounfidim}
    \sup_{t \in \R}\|{\rme^{\rmi t(H+\varepsilon V)}S\rme^{-\rmi t(H+\varepsilon V)}-S}\| \leq  \frac{14 \sqrt{d}\, \|V\|\|S\|}{\eta}   \varepsilon ,
\end{equation}
where $d$ is the number of distinct eigenvalues of $H$ and $\eta$ is the minimal spectral gap of $H$, i.e.,
\begin{equation}
    \eta=\min_{k\neq\ell}|{h_k-h_\ell}|,
\end{equation}
where the $h_k$'s are the distinct eigenvalues of $H$~\cite{kam}. 

In this paper, we want to generalize this result to  infinite-dimensional Hilbert spaces, including more general perturbations of the Hamiltonian $H$. Two questions arise in the effort of generalizing this kind of result to unbounded operators:
\begin{enumerate}
    \item [(i)] Given $S$ a robust symmetry against a perturbation $H(\varepsilon)$ of $H$ and a unit vector $\psi\in\mathcal{H}$. Is it possible to find a constant $C_{\psi}>0$ such that
    \begin{equation}\label{speed1}
        \sup_{t \in \R}\|{(\rme^{\rmi tH(\varepsilon)}S\rme^{-\rmi tH(\varepsilon)}-S)\psi}\|\leq C_{\psi} \varepsilon,
    \end{equation}
    for $\varepsilon$ sufficiently small? In other words, is the wandering range on state~$\psi$ of $O( \varepsilon )$?
    \item [(ii)] The constant $C$ is $\psi$-independent? In other words, is it true that
\begin{equation}\label{eq:normtop}  \sup_{t\in\mathbb{R}}\|{\rme^{{\rmi}tH(\varepsilon)}S\rme^{-{\rmi}tH(\varepsilon)}-S}\|\leq C \varepsilon ?
\end{equation} 
\end{enumerate}

The answer to both questions is negative (if $\mathcal{H}$ is infinite-dimensional).
There exist states $\psi$ for which the wandering range is of order $O( \varepsilon ^{\gamma})$ with $\gamma >0$ arbitrarily small. Moreover, in general, the convergence~\eqref{eq:rob} is not uniform in $\psi$ (with $\|\psi\|= 1$), and therefore it does not hold in operator norm.  We now present a simple explicit example~\cite{rob} that clearly illustrates this phenomenon.
\begin{example}
Let $H$ be the Hamiltonian of a one-dimensional harmonic oscillator with $m=\omega=1$ on the Hilbert space $L^2(\R)$, 
\begin{equation}
   H=\frac{1}{2}(p^2+x^2),
\end{equation}
and consider the following symmetry of $H$,
\begin{equation}
   S=\frac{1-\Pi}{2},
\end{equation}
where $\Pi$ is the parity operator, $\Pi \psi(x) =\psi(-x)$.

The symmetry $S$ can be proven~\cite{rob} to be robust against \emph{every} linear perturbation $H(\varepsilon) = H + \varepsilon V$ of $H$, with $V$ being Hermitian and $H$-bounded, so its wandering range vanishes $\delta_{H(\varepsilon)}(S,\psi)\to 0$ as $\varepsilon\to 0$, for all Hermitian and $H$-bounded $V$, and for all wave functions $\psi\in L^2(\mathbb{R})$.

However, by taking $V=p$ and 
\begin{equation}
   \psi_{\alpha}(x)=\frac{1}{(1+x^2)^{\alpha/4}},\quad x \in \R,
\end{equation}
with $\alpha >1$, 
it can be proven that, for all $\varepsilon \in [0,1]$,
\begin{equation}
   \delta_{H(\varepsilon)} (S, \psi_\alpha) =\sup_{t \in \R} \|{(\rme^{{\rmi}t(H + \varepsilon p)}S\rme^{{-\rmi}t (H + \varepsilon p)}-S)\psi_{\alpha}}\|\ge{c_\alpha}\varepsilon^{\frac{\alpha-1}{2}}.
 \end{equation}
Therefore, \eqref{speed1} cannot hold if $\alpha \in (1,3)$, and in fact, the speed of convergence can be arbitrarily slow for $\alpha\approx 1$.
Moreover, by taking $\alpha\downarrow 1$, one also gets that the convergence cannot be in operator norm, since   
\begin{equation}
   \sup_{t \in \R} \|{\rme^{{\rmi}t(H+\varepsilon p)}S\rme^{{-\rmi}t(H +\varepsilon p)}-S}\|\ge\frac{1}{\sqrt{2}}.
\end{equation}
\end{example}
The aim of this paper is to establish sufficient conditions on the perturbation $H(\varepsilon)$ of an unbounded Hamiltonian $H$, on the vector $\psi$, and on the robust symmetry $S$, such that the wandering range of $S$ is linear in $\varepsilon$ as in~\eqref{speed1}. We will also provide sufficient conditions under which the strong topology can be replaced by the uniform topology as in~\eqref{eq:normtop}, in analogy with the finite-dimensional setting.

Let us begin by introducing a special class of perturbations that will be central to our analysis.
These perturbations preserve the spectral structure of $H$ in a controlled way and will be referred to as \emph{admissible}.
\begin{definition}\label{def:addef}
\emph{(Admissible perturbation)}  
Let $H$ be a self-adjoint operator with pure point spectrum, and let $H(\varepsilon)$ be a perturbation of $H$.  
We say that $H(\varepsilon)$ is an \emph{admissible perturbation} of $H$ if for $\varepsilon\in[0,1]$ the following conditions hold:
\begin{enumerate}
    \item $H(\varepsilon)$ has a pure point spectrum. Therefore, for all $\psi \in D(H)$,
    \begin{equation}
    H(\varepsilon) \psi = \sum_{n \geq 1} h_n(\varepsilon) P_n(\varepsilon) \psi,
    \label{eq:spectralHeps}
    \end{equation}
    where $h_n(\varepsilon) \in \mathbb{R}$ and $P_n(\varepsilon)$ are the eigenvalues and spectral projections of $H(\varepsilon)$, respectively.
    \item For each $n \geq 1$, the map $\varepsilon \mapsto h_n(\varepsilon)$ is continuous.
    \item For $n \neq m$, we have $h_n(\varepsilon) \neq h_m(\varepsilon)$ for $\varepsilon \neq 0$.
    \item There exists a unitary operator $U(\varepsilon)$
    such that  
    \begin{equation}
    P_n(\varepsilon) = U(\varepsilon) P_n(0) U(\varepsilon)^\dagger, \quad \mathrm{for\ all}\ n \geq 1,
    \end{equation}
    and the map $\varepsilon \mapsto U(\varepsilon)$ is strongly continuous.
\end{enumerate}
\end{definition}
\begin{example}[Linear perturbation]
Consider the linear perturbation of $H$ of Example~\ref{exampleLD},
\begin{equation}
H(\varepsilon)=H+ \varepsilon V,  \quad \varepsilon \in [0,1].
\end{equation} 
If $H$ has compact resolvent, then  $H(\varepsilon)$ can be proven~\cite{rob} to be an admissible perturbation of $H$ (see also Example~\ref{lineardeformationdiff} for further details).
\end{example}
We recall here the main result of Ref.~\cite{rob}, i.e.~the algebraic characterization of  the robust symmetries of $H$ against its admissible perturbations.
\begin{theorem}[Robust symmetries] \label{thm:main0}
Let $S\in\{H\}'$ be a symmetry of $H$ and let $H(\varepsilon)$ be an admissible perturbation of $H$. $S$ is robust against $H(\varepsilon)$ if and only if $[S, P_n(0)]=0$ for all $n \geq 1$.
\end{theorem}
Another relevant result of Ref.~\cite{rob} is the algebraic characterization of the \emph{completely robust symmetries}, i.e.~the conserved quantities which are robust against all the possible (admissible) perturbations. This result is summarized by the following Theorem.
\begin{theorem}[Completely robust symmetries]\label{th:completely_robust}
    Let $S\in\{H\}'$. Then, $S$ is completely robust if and only if $S\in\{H\}''$, where
    \begin{equation}
        \{H\}''=\{S\in B(\mathcal{\mathcal{H}}):[S,A]=0,\ \forall A\in\{H\}'\}
    \end{equation}
    is the bicommutant of the Hamiltonian.
\end{theorem}
In other words, the completely robust symmetries are the operators that commute not only with the Hamiltonian, but also with all the symmetries of $H$. 

In the next section, we examine the wandering range $\delta_{H(\varepsilon)}(S,\psi)$ under admissible perturbations in two situations: when $\psi$ is a linear combination of eigenvectors of the unperturbed Hamiltonian $H$, and/or when the symmetry $S$ is of finite rank.
\section{Eigenvectors of the Hamiltonian and finite-rank symmetries}\label{sec:eigenvector}
Let $H$ be a self-adjoint operator on $\mathcal{H}$ with pure point spectrum. Its spectral decomposition reads
\begin{equation}\label{eq:spH}
    H\psi=\sum_{k\ge1}h_kP_k\psi, \quad \forall\psi\in{D(H)},
\end{equation}
where $\{h_k\}_{k\ge1}\subset\R$ are the distinct eigenvalues of $H$ and $\{P_k\}_{k\ge1}$ are its spectral projections,
\begin{equation}
    P_k=P_k^\dagger, \quad P_kP_\ell=\delta_{k\ell}P_k, \quad \forall k,\ell \geq 1, \qquad 
    \sum_{k\ge1}P_k\psi=\psi, \quad \forall\psi\in\mathcal{H}.
\end{equation}

Since $H$ has pure point spectrum, it admits an orthonormal basis of eigenvectors, whose linear span is a dense subspace $D$ of $\mathcal{H}$. Furthermore, the spectral projections $P_n(\varepsilon)$ in~\eqref{eq:spectralHeps} associated with an admissible perturbation $H(\varepsilon)$ of $H$  form a family of subprojections of the eigenprojections $P_k$ of $H$.

Given an admissible perturbation $H(\varepsilon)$ of $H$, the following theorem will establish sufficient conditions on a vector $\psi$ and/or a robust symmetry $S$, for the wandering range to be of order $\varepsilon$ in the strong or in the norm topology. 
\begin{theorem}\label{th:linear}
Let $H(\varepsilon)$ be an admissible perturbation of $H$ satisfying the following property: for all $k\geq 1$ there exists $c_k>0$ such that
\begin{equation}\label{eq:diff}
    \|(U(\varepsilon)-\mathbb{I})P_k\| \leq c_k  \varepsilon ,
\end{equation}
for $\varepsilon\in[0,\varepsilon_0]$, with $\varepsilon_0>0$. In other words, $U(\varepsilon)$ is  Lipschitz continuous on each eigenspace $\Range(P_k)$ of $H$.

Let $S \in \{H\}'$ be a robust symmetry against $H(\varepsilon)$. Then, for all $\varepsilon\in[0,\varepsilon_0]$ the following statements hold:
\begin{enumerate}
 \item[(i)] If $\psi \in \mathcal{H}$ belongs to the linear span $D$ of the eigenvectors of $H$, then there exists a
 $C_\psi > 0$ such that 
\begin{equation}\label{speedeig}
    \sup_{t \in \R} \|{(\rme^{\rmi t H(\varepsilon)} S \rme^{-\rmi t H(\varepsilon)} - S) \psi}\| \leq C_\psi  \varepsilon ;
\end{equation}
\item[(ii)] If $S$ is a finite-rank operator, then there exists a
$C > 0$ such that 
\begin{equation}\label{speedfr}
    \sup_{t\in \R} \|{\rme^{\rmi t H(\varepsilon)} S \rme^{-\rmi t H(\varepsilon)} - S}\| \leq C  \varepsilon .
\end{equation}
\end{enumerate}
\end{theorem}
\begin{proof} 
Let us introduce the following family of self-adjoint operators,
\begin{equation}
    \tilde{H}(\varepsilon)=U(\varepsilon)^\dagger H(\varepsilon)U(\varepsilon)=\sum_{n\ge1}h_n(\varepsilon)P_n(0), \quad \varepsilon \in [0,\varepsilon_0].
\end{equation}
 Since $S$ is robust against $H(\varepsilon)$, according to  Theorem~\ref{thm:main0}, $[S, P_m(0)]=0$ for all $m \geq 1$. Hence, $[S, \rme^{{\rmi}t\tilde{H}(\varepsilon)}]=0$ for all $t \in \R$ and $\varepsilon \in [0,\varepsilon_0]$.
     Therefore,   by recalling that the $P_n(0)$’s form a family of subprojections of the $P_k$’s, one gets,
         for all $t \in \R$ and $\varepsilon \in [0,\varepsilon_0]$,
\begin{eqnarray}
\fl \qquad    \| (\rme^{{\rmi}tH(\varepsilon)}S\rme^{-{\rmi}tH(\varepsilon)}-S)P_k \|&=& \|  \rme^{{\rmi}tH(\varepsilon)}[ S, \rme^{{-\rmi}tH(\varepsilon)}] P_k \| \nonumber  \\
    &=&\|[ S, \rme^{{-\rmi}tH(\varepsilon)} -\rme^{{-\rmi}t\tilde{H}(\varepsilon)}] P_k\| \nonumber \\
    & \leq & 2 \|S\| \| (\rme^{{-\rmi}tH(\varepsilon)} -\rme^{{-\rmi}t\tilde{H}(\varepsilon)})P_k\| \nonumber\\
    & = & 2 \|S\|\|( U(\varepsilon) \rme^{{-\rmi}t\tilde{H}(\varepsilon)} U(\varepsilon)^\dagger- \rme^{{-\rmi}t\tilde{H}(\varepsilon)} )P_k\| \label{use21}\\
    & = & 2 \|S\| \| U(\varepsilon)[\rme^{{-\rmi}t\tilde{H}(\varepsilon)}, U(\varepsilon)^\dagger ]P_k \| \nonumber\\
    & = & 2 \|S\| \|[\rme^{{-\rmi}t\tilde{H}(\varepsilon)}, U(\varepsilon)^\dagger - \mathbb{I} ]P_k \| \nonumber\\
    & = & 2 \|S\| \| [\rme^{{-\rmi}t\tilde{H}(\varepsilon)}, (U(\varepsilon)^\dagger - \mathbb{I}) P_k]\| \label{use22}\\
    & = & 2 \|S\| \| [U(\varepsilon)^\dagger (U(\varepsilon) - \mathbb{I}) P_k,\rme^{{-\rmi}t\tilde{H}(\varepsilon)} ] \| \nonumber\\
    & \leq & 4  \|S\| \| (U(\varepsilon)- \mathbb{I} )P_k\|\nonumber\\
    & \leq & 4\|S\|c_k\varepsilon, \label{eqnuseful}
\end{eqnarray}
where we used the unitary invariance of the norm,
the relation
\begin{equation}    
{H}(\varepsilon)=U(\varepsilon)\tilde{H}(\varepsilon)U(\varepsilon)^\dagger
\end{equation}
in~\eqref{use21}, 
the commutativity $[\rme^{{-\rmi}t\tilde{H}(\varepsilon)}, P_k]=0$ in~\eqref{use22},
and finally in~\eqref{eqnuseful} we  applied~\eqref{eq:diff}.

Now, we prove $(i)$. By assumption, 
\begin{equation}
  \psi = \sum_{j=0}^d P_{k_j} \psi,  
\end{equation}
 for some finite integer $d$.
Then, for all $\varepsilon\in[0,\varepsilon_0]$,
\begin{equation}
     \sup_{t \in \R} \|{(\rme^{{\rmi}tH(\varepsilon)}S\rme^{-{\rmi}tH(\varepsilon)}-S)\psi}\|\leq\sum_{j=1}^{d}\|{(\rme^{{\rmi}tH(\varepsilon)}S\rme^{-{\rmi}tH(\varepsilon)}-S)P_{k_j}\psi}\|.
\end{equation}
Hence, the bound~\eqref{speedeig} follows immediately by~\eqref{eqnuseful}, with
\begin{equation}
    C_\psi=4\|S\|\|\psi\|\sum_{j=1}^{d}c_{k_j}.
\end{equation}

Now, we prove $(ii)$. Since $S$ is a finite-rank  symmetry, then there is a finite integer $d$ such that
\begin{equation}
    S= \sum_{m=1}^d P_{k_m} S P_{k_m}.
\end{equation}
Therefore, for all $\varepsilon \in [0,\varepsilon_0]$,
\begin{eqnarray}
  \sup_{t \in \R} \|{(\rme^{{\rmi}tH(\varepsilon)}S\rme^{-{\rmi}tH(\varepsilon)}-S)}\| 
 & \leq  2 \|S\| \sum_{m=1}^d  \sup_{t \in \R} \|{( \rme^{-{\rmi}tH(\varepsilon)} - \rme^{-{\rmi}t\tilde{H}(\varepsilon)})P_{k_m}}\|  \nonumber\\
 & \leq 4 \|S\| \sum_{m=1}^d\| (U(\varepsilon)- \mathbb{I})P_{k_m}\|,
 \end{eqnarray}
hence~\eqref{speedfr} follows immediately by~\eqref{eq:diff}, with
\begin{equation}
    C=4\|{S}\|\sum_{m=1}^d c_{k_m}.
\end{equation}
\end{proof}

\subsection{\texorpdfstring{Linear perturbation $H(\varepsilon)= H+ \varepsilon V$}{Linear perturbation H(varepsilon)=H+varepsilon V}}\label{lineardeformationdiff}
Now, we want to present an application of Theorem~\ref{th:linear}\@. First of all, we recall a central result in perturbation theory due to Kato~\cite{Kato1995}.
\begin{theorem}\label{Katoth} Let $H$ be a self-adjoint operator with compact resolvent and spectral resolution~\eqref{eq:spH}. Let $V$ be a Hermitian $H$-bounded operator, with relative bound less than $1$.
Then,
\begin{enumerate}
\item for all $\varepsilon \in [0,1]$, the operator $H(\varepsilon)=H+ \varepsilon V$ is self-adjoint with domain $D(H(\varepsilon))=D(H)$, has compact resolvent, and its spectral decomposition reads
\begin{equation}
H(\varepsilon) \psi= \sum_{n \geq 1}h_n(\varepsilon)P_n(\varepsilon) \psi, \quad \psi \in D(H),
\label{eq:spectraldecthm}
\end{equation}
where $\{h_n(\varepsilon)\}_{n \geq 1}$ are the eigenvalues of $H(\varepsilon)$ and $\{P_n(\varepsilon)\}_{n \geq 1}$ are its finite-rank eigenprojections;
\item for all $n \geq 1$, the maps $ \varepsilon \in [0,1] \mapsto h_n(\varepsilon) \in \R$ and $\varepsilon \in [0,1] \mapsto P_n(\varepsilon) \in  B(\mathcal{H})$ are
analytic, with $h_n(\varepsilon)\neq h_m(\varepsilon)$ for $n\neq  m$ and $\varepsilon\neq0$;
\item the family $\{P_n(0)\}_{n \geq 1}$ is a family of subprojections of $\{P_k\}_{k \geq 1}$, namely, for all $n \geq 1$
there is a unique $k \geq 1$ such that $P_n(0)P_k = P_kP_n(0) = P_n(0)$, so that $\Range(P_n(0)) \subset \Range(P_k )$.
\end{enumerate}
\end{theorem}
We consider the {linear perturbation} of $H$,
\begin{equation}\label{lineardeformation1}
H(\varepsilon)=H+ \varepsilon V,  \quad \varepsilon \in [0,1],
\end{equation} 
 with $H$ and $V$ being as in Theorem~\ref{Katoth}\@. By making use of Theorems~\ref{th:linear} and~\ref{Katoth}, we get the following result.
\begin{theorem}
\label{thm:compactres}
Let $H$ be a self-adjoint operator with compact resolvent and spectral resolution~\eqref{eq:spH}. Let $V$ be a symmetric $H$-bounded operator.
Let $S \in \{H\}'$ be a robust symmetry against the linear  perturbation $H(\varepsilon)$ in~\eqref{lineardeformation1}.  Then, there exists $\varepsilon_0\in(0,1]$ such that for all $\varepsilon\in[0,\varepsilon_0]$ the following propositions are true:
\begin{enumerate}
\item[(i)] If $\psi \in \mathcal{H}$ belongs to the linear span $D$ of the eigenvectors of $H$,  then there exists a
$C_\psi > 0$ such that 
\begin{equation}\label{speedeig1}
    \sup_{t \in \R} \|{(\rme^{\rmi t (H+\varepsilon V)} S \rme^{-\rmi t (H+\varepsilon V)} - S) \psi}\| \leq C_\psi  \varepsilon ;
\end{equation}
\item[(ii)] If $S$ is a finite-rank operator, then there exists a 
$C > 0$ such that
\begin{equation}\label{speedfr1}
    \sup_{t\in \R} \|{\rme^{\rmi t (H+\varepsilon V)} S \rme^{-\rmi t (H+\varepsilon V)} - S}\| \leq C  \varepsilon .
\end{equation}
\end{enumerate}
\end{theorem}
\begin{proof}
The strategy of the proof is to show that the assumptions of Theorem~\ref{th:linear} are satisfied. In particular, we will show that $H(\varepsilon)$ is an admissible perturbation of $H$ and that property~\eqref{eq:diff} is verified for $\varepsilon\in[0,\varepsilon_0]$, for some $\varepsilon_0>0$.
In fact, according to Theorem~\ref{Katoth}, for all $\varepsilon \in [0,1]$ the operator $H(\varepsilon)$ has compact resolvent, with spectral decomposition~\eqref{eq:spectraldecthm}.

Moreover, $\varepsilon \mapsto h_n(\varepsilon)$ and $\varepsilon  \mapsto P_n(\varepsilon)$ are analytic functions, then conditions $(i)$--$(iii)$ of Definition~\ref{def:addef} are satisfied. 
Now, we show  that condition $(iv)$ of Definition~\ref{def:addef} is also true. 
Following Kato~\cite[sections 4.6 and 6.8]{Kato1995}, we define  
the operator $R_n(\varepsilon)=(P_n(\varepsilon)-P_n(0))^2$, for $n\ge1$. The operator $(\mathbb{I} - R_n (\varepsilon))^{-1/2}$ is well defined for $\|R_n(\varepsilon)\|<1$ and then   
the operator 
\begin{equation}
    U(\varepsilon)=\sum_{n\ge1}(\mathbb{I}-R_n(\varepsilon))^{-1/2}P_n(\varepsilon)P_n(0),
\end{equation}
is well defined and analytic for $\varepsilon$ sufficiently small. Moreover, $U(\varepsilon)$ is a family of unitary operators and $P_n(\varepsilon)=U(\varepsilon)P_n(0)U(\varepsilon)^\dagger$.

Now, we prove property~\eqref{eq:diff}. One has
\begin{eqnarray}
    \|{(U(\varepsilon)-\mathbb{I})P_n(0)}\|&=&\|(\mathbb{I}-R_n(\varepsilon))^{-1/2}P_n(\varepsilon)P_n(0)-P_n(0)\|\nonumber\\
    &\leq&\|{(\mathbb{I}-R_n(\varepsilon))^{-1/2}P_n(\varepsilon)P_n(0)-P_n(\varepsilon)P_n(0)}\|\nonumber\\
    &&{}+\|{P_n(\varepsilon)P_n(0)-P_n(0)}\|\nonumber\\
    &\leq&\|{P_n(\varepsilon)-P_n(0)}\|+\|{(\mathbb{I}-R_n(\varepsilon))^{-1/2}-\mathbb{I}}\|,
\end{eqnarray}
moreover, since the map $\varepsilon \mapsto P_n(\varepsilon)$ is analytic,  there exist $d_n>0$ such that
\begin{equation}
    \|{P_n(\varepsilon)-P_n(0)}\|\leq d_n\varepsilon
\end{equation}
and 
\begin{equation}
    \|{R_n(\varepsilon)}\|= \|{(P_n(\varepsilon)-P_n(0))^2}\|\leq d_n^2\varepsilon^2.
\end{equation}
Moreover,
 \begin{eqnarray}
     \fl \qquad \|{(\mathbb{I}-R_n(\varepsilon))^{-1/2}-\mathbb{I}}\|&\leq&{\sum_{j\ge1}\abs{-1/2\choose j}\|{R_n(\varepsilon)}\|^j} = \sum_{j\ge1}{2j\choose j}\norm{\frac{R_n(\varepsilon)}{4}}^j \nonumber\\
     & \leq & \sum_{j\ge1}{2j\choose j}\left(\frac{d_n^2  \varepsilon ^2}{4}\right)^j =\frac{1-\sqrt{1-(d_n\varepsilon)^2}}{\sqrt{1-(d_n\varepsilon)^2}}\leq d^2_n\varepsilon^2 ,
 \end{eqnarray}
 for $\varepsilon$ sufficiently small.
Now, recall that the $P_n(0)$'s are subprojections of the finite-rank eigenprojections $P_k$'s, whence
$P_k = \sum_{j=1}^J P_{n_j}(0)$ for some integer $J$.
Thus, there exist $c_n>0$ and $\varepsilon_0\in(0,1]$ such that, for all  $\varepsilon \in [0,\varepsilon_0]$,
\begin{equation}
    \fl \qquad \|{(U(\varepsilon)-\mathbb{I})P_k}\| \leq \sum_{j=1}^J \|{(U(\varepsilon)-\mathbb{I})P_{n_j}(0)}\|\leq \sum_{j=1}^J d_{n_j} (1 + d_{n_j}\varepsilon)  \varepsilon \leq c_k\varepsilon.
\end{equation}
Therefore, we can apply Theorem~\ref{th:linear} and the claims follow.
\end{proof}
\section{Uniformly bounded perturbations}\label{sec:bounded}

In this section, we study the wandering range of completely robust symmetries, 
that is, the conserved quantities which remain robust against all admissible 
perturbations of the Hamiltonian. According to Theorem~\ref{th:completely_robust}, 
the completely robust symmetries are precisely the operators belonging to the 
bicommutant of $H$. In the following theorem, we show that, when the perturbations 
are bounded, the wandering range is of order $\varepsilon$ in the norm topology, that is uniformly in the input state $\psi$.

\begin{theorem}\label{th:bicomm}
Let $H$ be a self-adjoint operator with pure point spectrum $\{h_k\}_{k\ge1}$ and nonvanishing minimal spectral gap, i.e.
\begin{equation}
\eta = \inf_{k \neq \ell} |h_k - h_\ell| > 0 .
\end{equation}
Let $H(\varepsilon)=H+\varepsilon V(\varepsilon)$ with 
$\varepsilon\in [0,1]$ be a uniformly bounded self-adjoint perturbation of $H$, with bound
\begin{equation}\label{eq:uniform_bound}
v= \sup_{\varepsilon\in[0,1]}\| V(\varepsilon)\| < +\infty.
\end{equation}

If $S$ belongs to the bicommutant of $H$, i.e.
\begin{equation}
S \in \{H\}''=\{ A \in  B(\mathcal{H}) : [A,B]=0, \forall B \in \{H\}' \},
\end{equation}
then its wandering range is bounded by
\begin{equation}\label{eq:est}
\sup_{t \in \mathbb{R}}
\|
\rme^{\rmi t H(\varepsilon)} S \rme^{-\rmi t H(\varepsilon)} - S
\|
\le
\beta  \frac{ v }{\eta} \|S\| \varepsilon,\quad \mathrm{for} \quad 0\leq\varepsilon\leq\frac{\eta}{v \rho},
\end{equation}
where $\beta\approx15.3$ and $\rho\approx34.8$ are defined in Table~\ref{tab:constants}\@.
\end{theorem}
\begin{remark}
If the perturbation is linear, namely $H(\varepsilon)=H+\varepsilon V$ with 
$V=V^\dagger \in {B}(\mathcal{H})$, then $V(\varepsilon)=V$ and
$v=\|V\|$.
\vspace{1em}
\end{remark}
\begin{remark}
The perturbation $H(\varepsilon)$ considered in Theorem~\ref{th:bicomm} satisfies the continuity condition~\eqref{eq:strongcont}, and in fact it is norm continuous,
\begin{equation}
    \|H(\varepsilon) - H \| = \varepsilon \|V(\varepsilon)\| \leq v \varepsilon,
\end{equation}
for $\varepsilon\in[0,1]$. However, it is \emph{not} necessarily admissible (see Definition~\ref{def:addef}). Indeed, if $H$ has an eigenvalue with infinite degeneracy, the perturbed Hamiltonian $H+\varepsilon V$ may develop a continuous spectral bands.

In this sense, Theorem~\ref{th:bicomm} substantially broadens the scope of Theorem~\ref{th:completely_robust} by encompassing a  wider class of perturbations. Specifically, it establishes that every element in the bicommutant of a Hamiltonian $H$ with pure point spectrum and nonvanishing minimal spectral gap ---allowing, in particular, for infinitely degenerate eigenspaces--- is robust under arbitrary bounded perturbations, with the wandering range uniformly controlled at order $\varepsilon$.

Moreover, the norm estimate~\eqref{eq:est} extends to unbounded Hamiltonians and nonlinear perturbations the analogous finite-dimensional bound recalled in~\eqref{eq:bounfidim}. At the same time, it strengthens that result by eliminating any dependence on the number $d$ of distinct eigenvalues.

\vspace{1em}
\end{remark}
\begin{remark}
The exact values of $\beta$ and $\rho$ in~\eqref{eq:est} are
\begin{equation}
    \beta= \frac{16 \pi}{\sqrt{3}} \alpha ( \rme^{\frac{1}{2\alpha}} - 1 ), \qquad 
    \rho=\frac{4\pi\alpha}{\sqrt{3}},
\end{equation}
where $\alpha \approx 4.79$ is the unique solution of the transcendental equation
\begin{equation}
(\alpha + 1)( \rme^{\frac{2}{\alpha}} - 1) = 3 .
\end{equation}
For convenience, the definition of these constants, together with other relevant quantities entering the bounds, are gathered in Table~\ref{tab:constants}\@.
\vspace{1em}
\end{remark}
\begin{table}[t]
\centering
\begin{tabular}{c c c}
\hline
Symbol & Definition & Numerical Value\\
\hline
\rule{0pt}{0.5cm}
$\alpha$ &
$\displaystyle (\alpha+1)(\rme^{\frac{2}{\alpha}}-1)=3$&
$<4.80$ \\[0.4cm]

$\beta$ &
$\displaystyle 
\beta=\frac{16\pi\alpha}{\sqrt{3}}(\rme^{\frac{1}{2\alpha}}-1)
$ &
$<15.3$ \\[0.4cm]

$\rho$ &
$\displaystyle \rho=\frac{4\pi\alpha}{\sqrt{3}}$ &
$<34.8$ \\[0.4cm]

$v$ &
$\displaystyle v=\sup_{\varepsilon\in[0,1]} \|V(\varepsilon)\|$ &
$<\infty$

\\[0.4cm]

$\eta$ &
$\displaystyle \eta=\inf_{k,\ell \geq 1 : k\neq \ell} |h_k - h_\ell|$ &
$> 0$

\\[0.4cm]

$b$ &
$\displaystyle b=\frac{\pi}{\sqrt{3}}\frac{v}{\eta}$ &
$\geq 0$
\\ [0.4cm]

\hline
\end{tabular}
\caption{Symbols appearing in Theorem~\ref{th:bicomm}\@.}
\label{tab:constants}
\end{table}
\begin{remark}
The condition~\eqref{eq:uniform_bound} requires a uniform bound on the incremental ratio of the perturbation $H(\varepsilon)$. This does not necessarily imply that $H(\varepsilon)$ is differentiable at $\varepsilon=0$, as shown in the following example.
\vspace{1em}
\end{remark}

\begin{example}
\label{ex:regularity}
    Let $V = V^\dagger \in B(\mathcal{H})$, and consider the following perturbation of $H$,
    \begin{equation}
        H+\varepsilon V(\varepsilon) = \left\{
        \begin{array}{ll}
        H & \mathrm{if}\ \varepsilon = 0, \\
          H+\varepsilon\sin\!\left(\frac{1}{\varepsilon}\right) V & \mathrm{if}\ \varepsilon \neq 0.
        \end{array}
        \right.
        \label{eq:57}
    \end{equation}
    This perturbation is clearly uniformly bounded, since
    \begin{equation}
        \|{V(\varepsilon)}\|\leq\left|{\sin\!\left( \frac{1}{\varepsilon} \right)}\right| \|V\| \leq \|V\|.
    \end{equation}
    However, $H(\varepsilon)$ is not differentiable at $\varepsilon = 0$. Indeed, for $\varepsilon\neq 0$,
    \begin{equation}
        H'(\varepsilon) = \left( \sin\!\left( \frac{1}{\varepsilon} \right) - \frac{1}{\varepsilon} \cos\!\left( \frac{1}{\varepsilon} \right) \right) V,
    \end{equation}
    which diverges as $\varepsilon \to 0$. Thus, uniform boundedness does not imply differentiability of $H(\varepsilon)$.

    Notice that, more generally, there are no constraints on the regularity of  $\varepsilon\mapsto V(\varepsilon)$, which may even be  \emph{everywhere discontinuous}, as far as it remains uniformly bounded. For instance, $\sin\!\left(\frac{1}{\varepsilon}\right)$ in~\eqref{eq:57} can be replaced by the Dirichlet function $\bm{1}_{\mathbb{Q}}(\varepsilon)$, while Theorem~\ref{th:bicomm} continues to hold.
\end{example}

The main idea of the proof is the construction of a new Hamiltonian 
$H+\varepsilon\hat{V}(\varepsilon)$ which generates a dynamics close to the perturbed one and that commutes with $H$. This idea is encoded in the following theorem, which is of independent interest.

\begin{theorem}[Eternal block-diagonal approximation]\label{th:eternal}
Let $H$ be a self-adjoint operator with pure point spectrum $\{h_k\}_{k\ge1}$ and nonvanishing minimal spectral gap
\begin{equation}
\eta = \inf_{k \neq \ell} |h_k - h_\ell| > 0 .
\end{equation}
Let $H(\varepsilon)=H+\varepsilon V(\varepsilon)$ with $\varepsilon\in [0,1]$ be a uniformly bounded self-adjoint perturbation of $H$, with bound
\begin{equation}
v= \sup_{\varepsilon\in[0,1]}\| V(\varepsilon)\| < +\infty.
\end{equation}

Then, there exists a family of bounded operators $\hat{V}(\varepsilon)$ such that
\begin{enumerate}
    \item $\hat{V}(\varepsilon)=\hat{V}(\varepsilon)^\dagger \in \{H\}'$;
    \item for all $t\in\mathbb{R}$,
    \begin{equation}\label{eternal_bound}
        \|
        \rme^{-\rmi t(H+\varepsilon V(\varepsilon))}
        -
        \rme^{-\rmi t(H+\varepsilon\hat{V}(\varepsilon))}
        \|
        \leq
        \frac{\beta}{2}\frac{v}{\eta}\varepsilon,
        \quad \mathrm{for} \quad 0\leq\varepsilon\leq\frac{\eta}{v \rho},
    \end{equation}
    where $\rho$ and $\beta$ are defined in Table~\ref{tab:constants}\@.
\end{enumerate}
The operator $H+\varepsilon\hat{V}(\varepsilon)$ is called the \emph{eternal block-diagonal approximation} of the perturbed Hamiltonian $H+\varepsilon V(\varepsilon)$.
\end{theorem}
\begin{remark}
   Since $H+\varepsilon\hat{V}(\varepsilon)$ commutes with the unperturbed Hamiltonian $H$, it is block-diagonal in the eigenbasis of $H$. 
    Moreover, the above estimate shows that the corresponding dynamics approximates the perturbed evolution uniformly in time, which motivates the term \emph{eternal block-diagonal} approximation~\cite{Eternal}. 
\end{remark}
We postpone the proof of Theorem~\ref{th:eternal} to the end of this section. Assuming its validity, we can now prove Theorem~\ref{th:bicomm}\@.
\begin{proof}[Proof of Theorem~\ref{th:bicomm}]
    Take $S\in\{H\}''$. Then, we have
    \begin{eqnarray}
    \|{\rme^{\rmi t(H+\varepsilon V(\varepsilon))}S\rme^{-\rmi t(H+\varepsilon V(\varepsilon))}-S}\|&=\|{ \rme^{\rmi t(H+\varepsilon V(\varepsilon))}[S,\rme^{-\rmi t(H+\varepsilon V(\varepsilon))}]}\|\nonumber\\
    &=\|{[S,\rme^{-\rmi t(H+\varepsilon V(\varepsilon))}-\rme^{-\rmi t(H+\varepsilon\hat{V}(\varepsilon))}]}\|\label{step_rob}\\
    &\leq2\|{S}\|\|{\rme^{-\rmi t(H+\varepsilon V(\varepsilon))}-\rme^{-\rmi t(H+\varepsilon\hat{V}(\varepsilon))}}\|\nonumber\\
    &\leq\frac{\beta\|{S}\| v }{\eta}\varepsilon,
    \label{step-bound}
\end{eqnarray}
where at step~\eqref{step_rob} we have used the fact that $S\in\{H\}''$, i.e.~$S$ commutes with all the operators in $\{H\}'$, in particular with $\hat{V}(\varepsilon)$; and at step~\eqref{step-bound} we  used the bound~\eqref{eternal_bound} on the evolutions of Theorem~\ref{th:eternal}\@.
\end{proof}
In order to construct the block-diagonal approximation appearing in Theorem~\ref{th:eternal}, we conjugate the perturbed Hamiltonian with a unitary operator close to the identity. This allows us to transform $H+\varepsilon V(\varepsilon)$ into an operator of the form $H+\varepsilon\hat V(\varepsilon)$, where $\hat V(\varepsilon)$ commutes with $H$. The precise statement is given in the following theorem.
\begin{theorem}\label{thm:fund}
Let $H$ be a self-adjoint operator with pure point spectrum $\{h_k\}_{k\ge1}$ and nonvanishing minimal spectral gap
\begin{equation}
\eta = \inf_{k \neq \ell} |h_k - h_\ell| > 0 .
\end{equation}
Let $H(\varepsilon)=H+\varepsilon V(\varepsilon)$ with $\varepsilon\in [0,1]$ be a uniformly bounded self-adjoint perturbation of $H$, with bound
\begin{equation}
v= \sup_{\varepsilon\in[0,1]}\| V(\varepsilon)\| < +\infty.
\end{equation}

Then, there exist two families of bounded operators, $W(\varepsilon)$ and $\hat{V}(\varepsilon)$, such that
\begin{enumerate}
    \item $\hat{V}(\varepsilon)=\hat{V}(\varepsilon)^\dagger \in \{H\}'$;
    \item $W(\varepsilon)$ is unitary;
    \item $W(\varepsilon)D(H)\subseteq D(H)$;
    \item $W(0)=\mathbb{I}$;
    \item for all $\psi\in D(H)$,
    \begin{equation}
    W(\varepsilon)^\dagger(H+\varepsilon V(\varepsilon))W(\varepsilon)\psi
    =
    (H+\varepsilon\hat{V}(\varepsilon))\psi ;
    \label{eq:main}
    \end{equation}
    \item $W(\varepsilon)$ obeys the bound
    \begin{equation}\label{ordereps}
        \|W(\varepsilon)-\mathbb{I}\|
        \leq
        \frac{\beta}{4}\frac{v}{\eta}\varepsilon,
       \quad \mathrm{for} \quad 0\leq\varepsilon\leq\frac{\eta}{v \rho}, 
    \end{equation}
    where $\rho$ and $\beta$ are defined in Table~\ref{tab:constants}\@.
\end{enumerate}
\end{theorem}

The proof of Theorem~\ref{thm:fund} will be the core of the next section, where we are going to explicitly construct the operators \( W(\varepsilon) \) and \( \hat{V}(\varepsilon) \).

Assuming the validity of Theorem~\ref{thm:fund}, we are able to prove Theorem~\ref{th:eternal}\@. More explicitly, we can prove that, if an operator $\hat{V}(\varepsilon)$ satisfies the hypotheses of Theorem~\ref{thm:fund}, then it is an eternal block-diagonal approximation, which satisfies the bound~\eqref{eternal_bound}.
\begin{proof}[Proof of Theorem~\ref{th:eternal}]
By using~\eqref{eq:main}, we have
\begin{eqnarray}
    \|{\rme^{-\rmi t(H+\varepsilon V(\varepsilon))}-\rme^{-\rmi t(H+\varepsilon\hat{V}(\varepsilon))}}\|&=\|{\rme^{-\rmi t(H+\varepsilon V(\varepsilon))}-W(\varepsilon)^\dagger \rme^{-\rmi t(H+\varepsilon V(\varepsilon))}W(\varepsilon)}\|\nonumber\\
    &=\|{W(\varepsilon)^\dagger[W(\varepsilon)-\mathbb{I},\rme^{-\rmi t(H+\varepsilon{V}(\varepsilon))}]}\|\nonumber\\
    &\leq2\|{W(\varepsilon)-\mathbb{I}}\|\nonumber\\
    &\leq \frac{\beta}{2}\frac{v}{\eta}\varepsilon, \label{step56}
\end{eqnarray}
where at step~\eqref{step56} we have used the bound~\eqref{ordereps} on the unitary $W(\varepsilon)$. 
\end{proof}

In the next section, we prove Theorem~\ref{thm:fund} by explicitly constructing the operators \( W(\varepsilon) \) and \( \hat{V}(\varepsilon) \) as formal series, and then establishing the convergence of those series.

\section{Formal construction and convergence analysis}\label{convergence}
\subsection{Homological equation}
In this section, we introduce a key lemma that deals with the so-called homological equation, also known as the commutator equation. 
This equation will play a central role in proving Theorem~\ref{thm:fund}\@. We will discuss its formal solution in terms of a series expansion, along with some important estimates.
Before presenting the lemma, let us introduce some useful notation.

Let $H$ be a self-adjoint operator with pure point spectrum and spectral resolution~\eqref{eq:spH}.
Given a bounded operator \( A \in B(\mathcal{H}) \), we define its block-diagonal part with respect to the Hamiltonian \( H \) as
\begin{equation}
    [A] \psi = \sum_{k \geq 1} P_k A P_k \psi,
\end{equation}
where the $P_k$'s are the eigenprojections of $H$, and the series converges for any $\psi\in\mathcal{H}$. Conversely, the off-diagonal part of $A$ is denoted by $\{A\}$, and defined as
\begin{equation}
    \{A\} \psi = A \psi - [A] \psi = \sum_{k,\ell\geq 1 : k \neq\ell} P_k A P_\ell \psi,
\end{equation}
for all $\psi\in\mathcal{H}$.

\begin{lemma}[Homological equation]\label{l:61}
Let $H$ be a self-adjoint operator with pure point spectrum 
$\{h_k\}_{k\ge1}$ and nonvanishing minimal spectral gap
\begin{equation}
\eta = \inf_{k \neq \ell} |h_k - h_\ell| > 0 .
\end{equation}
Let $B \in  B(\mathcal{H})$ be a bounded operator. Then, there exists a unique bounded operator
$X \in  B(\mathcal{H})$ satisfying the following conditions:
\begin{enumerate}
    \item[(i)] $ X D(H) \subseteq D(H) $;
    \item[(ii)] for all $ \psi \in D(H) $,
    \begin{equation}\label{eq:homological}
        \rmi [X, H] \psi = \{B\}\psi;
    \end{equation}
    \item[(iii)] the block-diagonal part of $ X $ vanishes, i.e.,
    \begin{equation}
        [X] = 0.
    \end{equation}
\end{enumerate}

The operator $X$ has the explicit form
\begin{equation}\label{eq:solution}
    X \psi = \rmi \sum_{k,\ell \geq 1 : k \neq \ell} \frac{P_k B P_\ell}{h_k - h_\ell} \psi,
\end{equation}
for all $\psi\in\mathcal{H}$.
Moreover, if $B$ is self-adjoint, then $X$ is self-adjoint.
Finally, the following estimates hold:
\begin{eqnarray}
    \|X\| \leq \frac{\pi}{\sqrt{3}\, \eta} \|B\|, \\
    \|HX \psi\| \leq \left( 2\|\psi\| + \frac{\pi}{\sqrt{3}\, \eta} \|H\psi\| \right) \|B\|, \quad \forall\psi \in D(H).
    \label{eq:relative}
\end{eqnarray}
Equation~\eqref{eq:homological} is known in the literature as \emph{homological equation}.
\end{lemma}
\begin{proof}
Let $ \phi \in \mathcal{H} $. Since~\eqref{eq:homological} must hold for all $ \psi \in D(H) $, it must in particular hold for $ \psi =  P_\ell \phi $, for any $ \ell \geq 1 $. In this case, we obtain
\begin{equation}
\rmi ( h_\ell X P_\ell - H X P_\ell) \phi = \{B\} P_\ell \phi.
\end{equation}
By left-multiplying both sides by $P_k$, with $ k \geq 1 $, we find that
\begin{equation}
(h_k - h_\ell) P_k X P_\ell \phi = \rmi P_k \{B\} P_\ell \phi.
\end{equation}
For $ k = \ell $, the equation is automatically satisfied for any choice of $ X $, while for $ k \neq \ell $ it gives
\begin{equation}
P_k X P_\ell \phi = \rmi \frac{P_k B P_\ell}{h_k - h_\ell} \phi.
\end{equation}
Summing over all $k \neq \ell$, we formally obtain the  identity
\begin{equation}
\{X\} \phi = \rmi \sum_{k,\ell : k \neq \ell} \frac{P_k B P_\ell}{h_k - h_\ell} \phi.
\end{equation}

Thus, the homological equation determines only the off-diagonal part of $ X $, leaving its block-diagonal part arbitrary. Therefore, the solution is unique up to the addition of a block-diagonal operator. By fixing the block-diagonal part to zero, i.e., by imposing $ [X] = 0 $, we select the unique solution
\begin{equation}
    X \phi = \rmi \sum_{k,\ell : k \neq \ell} \frac{P_k B P_\ell}{h_k - h_\ell} \phi.
\label{eq:X}
\end{equation}
Furthermore, if $B=B^\dagger$ it is straightforward to verify that (formally) $X=X^\dagger$.

Clearly, it is necessary to verify the convergence of the series defining $ X $. Let $ \psi \in \mathcal{H} $. Using the orthogonality of the projections $ P_k $, we obtain
\begin{equation}
\|{X\psi}\|^2 \leq \|{B}\|^2 \sum_{\ell \geq 1} \sum_{k\geq 1: k \neq \ell} \frac{\|{P_k \psi}\|^2}{|h_k - h_\ell|^2}.
\end{equation}
Exchanging the order of summations gives
\begin{equation}
\|{X\psi}\|^2 \leq \|{B}\|^2 \sum_{k \geq 1} \|{P_k \psi}\|^2 \sum_{\ell\geq 1 : \ell \neq k} \frac{1}{|h_k - h_\ell|^2}.
\end{equation}

We now use the fact that $H$ has a nonvanishing spectral gap $\eta>0$. 
By assuming that the eigenvalues $\{h_k\}_{k\ge1}$ are ordered increasingly, we have
\begin{equation}
|h_k - h_\ell| \geq \eta |k - \ell|.
\end{equation}
Indeed, without loss of generality, assume $k \ge \ell$. Then,
\begin{eqnarray}
|h_k - h_\ell| = h_k - h_\ell 
&= (h_k - h_{k-1}) + (h_{k-1} - h_{k-2}) + \cdots + (h_{\ell+1} - h_\ell) \nonumber \\
&\ge \eta + \eta + \cdots + \eta = (k-\ell)\eta = \eta |k-\ell|.
\end{eqnarray}
Applying this bound yields
\begin{equation}
\|{X\psi}\|^2 \leq \frac{\|{B}\|^2}{\eta^2} \sum_{k \geq 1} \|{P_k \psi}\|^2 \sum_{\ell\geq 1 : \ell \neq k} \frac{1}{(k - \ell)^2}.
\end{equation}
Changing variables in the second sum by setting $ j = \ell - k $, we find
\begin{eqnarray}
\|{X\psi}\|^2 &\leq& \frac{\|{B}\|^2}{\eta^2} \sum_{k \geq 1} \|{P_k \psi}\|^2 \left(\sum_{j = -k+1}^{-1} \frac{1}{j^2} + \sum_{j = 1}^{\infty} \frac{1}{j^2} \right) \nonumber \\
&=& \frac{\|{B}\|^2}{\eta^2} \sum_{k \geq 1} \|{P_k \psi}\|^2 \left( \sum_{j = 1}^{k - 1} \frac{1}{j^2} + \frac{\pi^2}{6} \right) \nonumber \\
& \leq & \frac{\|{B}\|^2}{\eta^2} \sum_{k \geq 1} \|{P_k \psi}\|^2 \left( \sum_{j = 1}^{+\infty} \frac{1}{j^2} + \frac{\pi^2}{6} \right) \nonumber\\
&=& \frac{2 \pi^2 \|{B}\|^2}{6 \eta^2} \sum_{k \geq 1} \|{P_k \psi}\|^2 = \frac{\pi^2 \|{B}\|^2}{3 \eta^2} \|{\psi}\|^2.
\end{eqnarray}
Hence,
\begin{equation}
\|{X}\| \leq \frac{\pi}{\sqrt{3}\, \eta} \|{B}\|. \label{eq:X1}
\end{equation}

We now prove the bound~\eqref{eq:relative}. Let $ \psi \in D(H) $.  Using the explicit expression~\eqref{eq:solution} for $ X $, we write
\begin{eqnarray}
\fl \qquad -\rmi H X \psi &=& \sum_{k,\ell\geq 1 : k \neq \ell} \frac{h_k}{h_k - h_\ell} P_k B P_\ell \psi 
= \sum_{k,\ell\geq 1 : k \neq \ell} \left( \frac{h_k - h_\ell + h_\ell}{h_k - h_\ell} \right) P_k B P_\ell \psi \nonumber\\
&=& \sum_{k,\ell\geq 1 : k \neq \ell} P_k B P_\ell \psi + \sum_{k,\ell\geq 1 : k \neq \ell} \frac{h_\ell}{h_k - h_\ell} P_k B P_\ell \psi = \{B\} \psi - \rmi X H \psi.
\end{eqnarray}
Taking the norm of both sides and applying the triangle inequality gives
\begin{equation}\label{eq:step79}
    \| H X \psi \| \leq \| \{B\} \psi \| + \| X \|  \| H \psi \|.
\end{equation}
Since
\begin{equation}
    \| \{B\} \| = \| B - [B] \| \leq \| B \| + \| [B] \|,
\end{equation}
and
\begin{equation}
    \| [B] \| = \sup_{k\ge1} \| P_k B P_k \| \leq \| B \|,
\end{equation}
we obtain
\begin{equation}
    \| \{B\} \| \leq 2 \| B \|.
\end{equation}
Substituting this bound into~\eqref{eq:step79} yields
\begin{equation}
    \| H X \psi \| \leq 2 \| B \|  \| \psi \| + \| X \|  \| H \psi \|.
\end{equation}
Finally, using the estimate~\eqref{eq:X1} for $ \| X \| $, we conclude
\begin{equation}
\| H X \psi \| \leq \left( 2 \| \psi \| + \frac{\pi}{\sqrt{3}\, \eta} \| H \psi \| \right) \| B \|.
\end{equation}
\end{proof}
\begin{remark}
The homological equation depends only on the off-diagonal part of the operator $ B $. Indeed, one can always replace $ B $ with $ B + [A] $, for any $ A \in  B(\mathcal{H}) $, without affecting the solution. This freedom allows for an optimization of the norm estimate of the solution $X$ as
\begin{equation}
\|X\| \leq \frac{\pi}{\sqrt{3}\, \eta} \inf_{A \in  B(\mathcal{H})} \| B + [A] \|.
\end{equation}
The infimum is not necessarily attained by choosing $ A = -[B] $, as it may happen that
\begin{equation}
\| B - [B] \| = \| \{ B \} \| > \| B \|.
\end{equation}
For instance, consider the matrix~\cite{bhatia1989offdiagonal}
\begin{equation}
B =
\left(
  \begin{array}{ccc}
    -{1 \over 2} & 1 & 1 \\
    1 & -{1 \over 2} & 1 \\
    1 & 1 & -{1 \over 2}
  \end{array}
\right),
\end{equation}
for which it is straightforward to verify that
\begin{equation}
\| \{ B \} \| = 2 > \| B \| = \frac{3}{2},
\end{equation}
where the off-diagonal decomposition is taken with respect to $1$-dimensional blocks, that is,
\begin{equation}
\{B\} =
\left(
  \begin{array}{ccc}
    0 & 1 & 1 \\
    1 & 0 & 1 \\
    1 & 1 & 0
  \end{array}
\right).
\end{equation}
\end{remark}

\subsection{Quantum KAM iteration: Formal construction}

We return to our main problem: we aim to construct two bounded operators $ W(\varepsilon) $ and $ \hat{V}(\varepsilon) $ satisfying the properties of Theorem~\ref{thm:fund}\@. In particular, we require that, for all $ \psi \in D(H) $,
\begin{equation}
    W(\varepsilon)^\dagger (H + \varepsilon V(\varepsilon)) W(\varepsilon) \psi 
    = 
    (H + \varepsilon \hat{V}(\varepsilon)) \psi.
    \label{eq:diageq}
\end{equation}
For convenience, from now on we omit the dependence on $\psi$, and the equations will be understood to hold on their appropriate domains.

We seek $ W(\varepsilon) $ in the form
\begin{equation}
    W(\varepsilon) = \rme^{\rmi K(\varepsilon)},
\end{equation}
where $ K(\varepsilon) = K(\varepsilon)^\dagger $ and $ K(\varepsilon) D(H) \subseteq D(H) $. 

It will be convenient to regard $K(\varepsilon)$ and $\hat V(\varepsilon)$ as functions of two variables, $K(\varepsilon,\xi)$ and $\hat V(\varepsilon,\xi)$, evaluated along the diagonal $\xi = \varepsilon$, namely,
\begin{equation}
    K(\varepsilon) = K(\varepsilon, \varepsilon), \qquad
    \hat V(\varepsilon) = \hat V(\varepsilon,\varepsilon).
\end{equation}
The operators $K(\varepsilon,\xi)$ and $\hat V(\varepsilon,\xi)$ are required to satisfy the equation
\begin{equation}
    \rme^{-\rmi K(\varepsilon,\xi)}(H + \varepsilon V(\xi))\rme^{\rmi K(\varepsilon,\xi)} = H + \varepsilon \hat{V}(\varepsilon,\xi)
    \label{unitequiv}
\end{equation}
on $D(H)$, which reduces to~\eqref{eq:diageq} for $\xi=\varepsilon$.

We look for $K(\varepsilon,\xi)$ and $\hat{V}(\varepsilon,\xi)$ as formal power series in $\varepsilon$, 
\begin{eqnarray}
    K(\varepsilon,\xi) = \sum_{s \geq 1} \varepsilon^s K_s(\xi), \label{eq:K} \\
    \hat{V}(\varepsilon,\xi) = \sum_{s \geq 0} \varepsilon^s \hat{V}_s(\xi), \label{eq:V}
\end{eqnarray}
with $K(0,\xi)=0$ ensuring that $W(0) = \rme^{\rmi K(0)} = \mathbb{I}$.
\begin{remark}
If the perturbation $V$ in~\eqref{unitequiv} depends on $\xi$, then the coefficients $K_s(\xi)$ and $\hat V_s(\xi)$ also inherit this dependence. This justifies the notation adopted in~\eqref{eq:K}–\eqref{eq:V}.
Recall that, as discussed in Example~\ref{ex:regularity}, there are no constraints on the regularity of the map $\xi\mapsto V(\xi)$, which may even be  everywhere discontinuous, as far as it remains uniformly bounded. Therefore, in general, the functions $K(\varepsilon,\xi)$ and $\hat{V}(\varepsilon,\xi)$  are not expected to have any regularity with respect to the second variable $\xi$, whereas the dependence on the first variable $\varepsilon$ will be shown to be analytic in the norm topology, uniformly for $\xi \in [0,1]$.
\end{remark}
We now derive the recursive conditions on the coefficients $ K_s(\xi) $ and $ \hat{V}_s(\xi) $ in order to satisfy~\eqref{unitequiv} order by order in $ \varepsilon$. This is obtained by inserting the formal series expansions~\eqref{eq:K} and~\eqref{eq:V} into both sides of the equation and comparing terms with the same power of $ \varepsilon $. This procedure is called quantum KAM iteration, since it recalls the analogous procedure used in the Kolmogorov-Arnold-Moser theory of classical mechanics~\cite{Arnold_DynSysIII_1988,Dumas}. The result is encoded in the following lemma.

\begin{lemma}\label{l:62}
    Let $H$ and $V(\varepsilon)$ be as in Theorem~\ref{thm:fund}\@.
    Assume that $ K(\varepsilon,\xi) $ and $ \hat{V}(\varepsilon,\xi) $ admit the expansions~\eqref{eq:K} and~\eqref{eq:V}. Then,~\eqref{unitequiv} holds on $D(H)$ if and only if, for all $ s \geq 1 $,
    \begin{equation}\label{eq:iter}
        \hat{V}_{s-1}(\xi) = B_s(\xi) - \rmi [K_s(\xi), H],
    \end{equation}
    where $B_s(\xi)$ is defined as
    \begin{eqnarray}
        B_1(\xi) &=& V(\xi), \nonumber\\
        B_s(\xi) &=& \sum_{n = 2}^s \frac{(-\rmi)^n}{n!} \sum_{\bm{\ell}\in\mathbb{N}^n : |\bm{\ell}| = s} \tilde{\mathcal{K}}^\xi_{\ell_1} \cdots \tilde{\mathcal{K}}^\xi_{\ell_n}(H) \nonumber \\
        &&{}+ \sum_{n = 1}^{s - 1} \frac{(-\rmi)^n}{n!} \sum_{\bm{\ell}\in\mathbb{N}^n : |\bm{\ell}| = s - 1} \tilde{\mathcal{K}}^\xi_{\ell_1} \cdots \tilde{\mathcal{K}}^\xi_{\ell_n}(V(\xi)),
        \label{eq:Bs}
    \end{eqnarray}
    where $\tilde{\mathcal{K}}^\xi_{\ell_k} (A)=[K_{\ell_k}(\xi),A]$ 
    for any suitable  operator $A$. We use the multi-index notation $ \bm{\ell} = (\ell_1, \dots, \ell_n) \in \mathbb{N}^n$, and $ |\bm{\ell}| = \ell_1 + \cdots + \ell_n $.
\end{lemma}

\begin{proof}
We insert the expansions~\eqref{eq:K} and~\eqref{eq:V} into~\eqref{unitequiv}, obtaining
\begin{eqnarray}
    \sum_{s \geq 0} \varepsilon^{s+1} \hat{V}_s(\xi) = \varepsilon V(\xi)&&{} + \sum_{s \geq 1} \varepsilon^s \sum_{n = 1}^s \frac{(-\rmi)^n}{n!} \sum_{|\bm{\ell}| = s} \tilde{\mathcal{K}}^\xi_{\ell_1} \cdots \tilde{\mathcal{K}}^\xi_{\ell_n}(H) \nonumber \\
    &&{}+ \sum_{s \geq 1} \varepsilon^{s+1} \sum_{n = 1}^s \frac{(-\rmi)^n}{n!} \sum_{|\bm{\ell}| = s} \tilde{\mathcal{K}}^\xi_{\ell_1} \cdots \tilde{\mathcal{K}}^\xi_{\ell_n}(V(\xi)).
    \label{equation}
\end{eqnarray}
Matching the coefficients of each power of $ \varepsilon $ we find at the first order
\begin{equation}
    \hat{V}_0(\xi) = V(\xi) - \rmi \tilde{\mathcal{K}}^\xi_1(H)
    = V(\xi) - \rmi [K_1(\xi), H],
    \label{first_order}
\end{equation}
which coincides with~\eqref{eq:iter} for $ s = 1$. At higher orders $ s \geq 2 $, we obtain
\begin{eqnarray}
    \hat{V}_{s-1}(\xi) &=& \sum_{n = 1}^s \frac{(-\rmi)^n}{n!} \sum_{|\bm{\ell}| = s} \tilde{\mathcal{K}}^\xi_{\ell_1} \cdots \tilde{\mathcal{K}}^\xi_{\ell_n}(H) \nonumber \\
    &&{} + \sum_{n = 1}^{s-1} \frac{(-\rmi)^n}{n!} \sum_{|\bm{\ell}| = s-1} \tilde{\mathcal{K}}^\xi_{\ell_1} \cdots \tilde{\mathcal{K}}^\xi_{\ell_n}(V(\xi)).
    \label{higher}
\end{eqnarray}
Isolating the term with $ n = 1 $ in the first sum, we write
\begin{equation}
    -\rmi [K_s(\xi), H] + \sum_{n = 2}^s \frac{(-\rmi)^n}{n!} \sum_{|\bm{\ell}| = s} \tilde{\mathcal{K}}^\xi_{\ell_1} \cdots \tilde{\mathcal{K}}^\xi_{\ell_n}(H),
\end{equation}
which, substituted into~\eqref{higher}, yields the recursive identity~\eqref{eq:iter}.
\end{proof}

We now solve~\eqref{eq:iter} by imposing the main requirement of the construction: $\hat{V}(\varepsilon,\xi) \in \{H\}'$. Imposing this condition order by order leads to explicit expressions for both $ \hat{V}_\ell(\xi) $ and $ K_\ell(\xi) $, as stated in the following result.

\begin{lemma}\label{l:63}
Equation~\eqref{eq:iter} in Lemma~\ref{l:62} admits a unique solution under the constraint $ [K_s(\varepsilon)] = 0 $, given by the formally self-adjoint operators
\begin{eqnarray}
    \hat{V}_{s-1}(\varepsilon) = \sum_{k \geq 1} P_k B_s(\varepsilon) P_k, \label{eq:block_s} \\
    K_s(\varepsilon) = \rmi \sum_{k \neq \ell} \frac{P_k B_s(\varepsilon) P_\ell}{h_k - h_\ell}, \label{eq:gen_s}
\end{eqnarray}
\end{lemma}
\begin{proof}
Fix $ s \geq 1 $. According to Lemma~\ref{l:62}, the operators $ \hat{V}_{s-1}(\varepsilon) $ and $ K_s(\varepsilon) $ must satisfy
\begin{equation}
    \hat{V}_{s-1}(\varepsilon) = B_s(\varepsilon) - \rmi [K_s(\varepsilon), H].
    \label{eq:iter1}
\end{equation}
We impose that $ \hat{V}_{s-1}(\varepsilon) \in \{H\}' $, i.e., that it commutes with $ H $. This is equivalent to requiring
\begin{equation}
    \hat{V}_{s-1}(\varepsilon) = [\hat{V}_{s-1}(\varepsilon)].
\end{equation}
Taking the block-diagonal part of both sides of~\eqref{eq:iter1}, and using that the block-diagonal part of a commutator with $ H $ vanishes, we find
\begin{equation}
    \hat{V}_{s-1}(\varepsilon) = [B_s(\varepsilon)] = \sum_{k \geq 1} P_k B_s(\varepsilon) P_k,
\end{equation}
which proves~\eqref{eq:block_s}.
Next, we determine $ K_s(\varepsilon) $ by taking the off-diagonal part of~\eqref{eq:iter1},
\begin{equation}
    \rmi [K_s(\varepsilon), H] = \{ B_s(\varepsilon) \}.
    \label{homological}
\end{equation}
This is a homological equation of the type considered in Lemma~\ref{l:61}, whose unique solution under the constraint $ [K_s(\varepsilon)] = 0 $ is given by
\begin{equation}
    K_s(\varepsilon) = \rmi \sum_{k \neq \ell} \frac{P_k B_s(\varepsilon) P_\ell}{h_k - h_\ell},
\end{equation}
as stated in~\eqref{eq:gen_s}.

It remains to prove that the operators $\hat{V}_{s-1}(\varepsilon)$ and $K_s(\varepsilon)$ are self-adjoint, in order to complete the argument. It is sufficient to show that $B_s(\varepsilon)$ is self-adjoint. We proceed by induction on $s$. For $s = 1$, we have
\begin{equation}
    B_1(\varepsilon) = V(\varepsilon),
\end{equation}
which is bounded and self-adjoint by hypothesis. Suppose now that
\begin{equation}
    B_j(\varepsilon) = B_j(\varepsilon)^\dagger, \quad \mathrm{for\ all}\ j = 1, \dots, s-1.
\end{equation}
Then, the corresponding operators $K_j(\varepsilon)$ are also self-adjoint for all $j = 1, \dots, s-1$ (Lemma~\ref{l:61}). Now consider the definition~\eqref{eq:Bs} of $B_s(\varepsilon)$: it is a linear combination of terms of the form
\begin{equation}
    (-\rmi)^n \tilde{\mathcal{K}}^\varepsilon_{\ell_1} \cdots \tilde{\mathcal{K}}^\varepsilon_{\ell_n}(A)=(-\rmi)^n[K_{\ell_1}(\varepsilon),\cdots,[K_{\ell_n}(\varepsilon),A]\cdots],
\end{equation}
where $A = H$ or $V$ and $\ell_i \leq s - 1$ for all $i=1,\dots, n$. Hence all these terms are self-adjoint operators.
\end{proof}
\subsection{Convergence of the formal expansion}
The construction presented in the previous subsection is purely formal. The operators $B_s(\xi)$, $K_s(\xi)$, and $\hat{V}_s(\xi)$ are defined through infinite series, and their convergence must still be proved. Similarly, the operators $K(\varepsilon,\xi)$ and $\hat{V}(\varepsilon,\xi)$ are introduced via formal expansions in $\varepsilon$, whose convergence remains to be established. It is therefore necessary to prove that these series converge in a nontrivial neighborhood of $\varepsilon = 0$. In order to prove the convergence of the series, we will make use of a well-known combinatorial sequence: the \emph{Catalan numbers}. These numbers are defined recursively by
\begin{equation}
    d_1 = 1; \qquad 
    d_s = \sum_{\ell=1}^{s-1} d_\ell d_{s-\ell}, \quad \mathrm{for}\quad s \geq 2.
    \label{eq:Catalan_numbers}
\end{equation}
\begin{remark}
Catalan numbers are usually indexed starting from $s=0$~\cite{stanley2015catalan}. 
Here, we adopt an equivalent convention starting from $s=1$, which is more convenient for our purposes. 
Further details on the Catalan numbers can be found in~\ref{appendix}\@.
\end{remark}
We begin by deriving explicit estimates for the operators $B_s(\varepsilon)$.
\begin{lemma}\label{l:2}
Let $H$ and $V(\varepsilon)$ be as in Theorem~\ref{thm:fund}\@.
Let $B_s(\varepsilon)$ be the operators defined in~\eqref{eq:Bs} of Lemma~\ref{l:62}. Then, for all $s \geq 1$,
\begin{equation}
    \frac{\pi}{\sqrt{3}\,\eta} \| B_s(\varepsilon) \| \leq \alpha^{s-1} b^s d_s,
    \label{second_bound}
\end{equation}
where $\eta>0$ is the minimal spectral gap of $H$,
\begin{equation}
    b = \frac{\pi v}{\sqrt{3}\,\eta} , \quad \mathrm{with} \quad v=\sup_{\varepsilon\in[0,1]}\|{V(\varepsilon)}\|,
\end{equation}
$\alpha$ is the solution of the transcendental equation
\begin{equation}\label{eq:alpha}
    (\alpha + 1)( \rme^{\frac{2}{\alpha}} - 1 ) = 3,
\end{equation}
and $d_s$, $s\geq 1$, are the Catalan numbers~\eqref{eq:Catalan_numbers}.
\end{lemma}

\begin{proof}
We proceed by induction on $s$. For $s = 1$, the bound~\eqref{second_bound} is trivially satisfied since
\begin{equation}
    \frac{\pi}{\sqrt{3}\,\eta} \| B_1(\varepsilon) \| = \frac{\pi}{\sqrt{3}\,\eta} \| V(\varepsilon) \| \leq\frac{\pi}{\sqrt{3}\,\eta} v= b = \alpha^{0} b^1 d_1 .
\end{equation}
Let $s \geq 2$ and assume that 
\begin{equation}
    \frac{\pi}{\sqrt{3}\,\eta} \| B_\ell(\varepsilon) \| \leq \alpha^{\ell-1} b^\ell d_\ell, \quad \mathrm{for\ all}\ \ell= 1, \dots, s-1.
\end{equation}
We have to prove that $ \frac{\pi}{\sqrt{3}\,\eta} \| B_s(\varepsilon) \| \leq \alpha^{s-1} b^s d_s$. From the definition~\eqref{eq:Bs} and using the triangle inequality, we obtain
\begin{eqnarray}
    \frac{\pi}{\sqrt{3}\,\eta} \| B_s(\varepsilon) \| 
    &\leq& \frac{\pi}{\sqrt{3}\,\eta} \sum_{n=2}^s \frac{1}{n!} \sum_{|\bm{\ell}|=s} 
    \| \tilde{\mathcal{K}}_{\ell_1}^{\varepsilon} \cdots \tilde{\mathcal{K}}_{\ell_n}^{\varepsilon}(H) \| \nonumber \\
    && {}+ \frac{\pi}{\sqrt{3}\,\eta} \sum_{n=1}^{s-1} \frac{1}{n!} \sum_{|\bm{\ell}|=s-1} 
    \| \tilde{\mathcal{K}}_{\ell_1}^{\varepsilon} \cdots \tilde{\mathcal{K}}_{\ell_n}^{\varepsilon}(V(\varepsilon)) \|.
    \label{eq:Ksm}
\end{eqnarray}
Consider  the first sum
\begin{equation}
    \frac{\pi}{\sqrt{3}\,\eta} \sum_{n=2}^s \frac{1}{n!} \sum_{|\bm{\ell}|=s} 
    \| \tilde{\mathcal{K}}_{\ell_1}^{\varepsilon} \cdots \tilde{\mathcal{K}}_{\ell_n}^{\varepsilon}(H) \|.
    \label{eq:firstsum}
\end{equation}
Recall that for all \( s \geq 1 \), \( K_s(\varepsilon) \) is the solution of the homological equation
\begin{equation}
    \rmi [K_s(\varepsilon), H] = \{ B_s(\varepsilon) \}.
\end{equation}
It follows that
\begin{equation}
    \| [K_{\ell_n}(\varepsilon), H] \| \leq 2 \| B_{\ell_n}(\varepsilon) \|,
\end{equation}
and hence
\begin{eqnarray}
\frac{\pi}{\sqrt{3}\,\eta} \| \tilde{\mathcal{K}}_{\ell_1}^{\varepsilon} \cdots \tilde{\mathcal{K}}_{\ell_n}^{\varepsilon}(H) \|
&\leq 2^n \| K_{\ell_1}(\varepsilon) \| \cdots \| K_{\ell_{n-1}}(\varepsilon) \|  \frac{\pi}{\sqrt{3}\,\eta} \| B_{\ell_n}(\varepsilon) \| \nonumber\\
&\leq 2^n \prod_{j=1}^{n} \left( \frac{\pi}{\sqrt{3}\,\eta} \| B_{\ell_j}(\varepsilon) \| \right), \label{usehomeq}
\end{eqnarray}
where in the last line we used the bound~\eqref{eq:relative} in Lemma~\ref{l:61}, that is \( \| K_{\ell_j}(\varepsilon) \| \leq \frac{\pi}{\sqrt{3}\, \eta}  \| B_{\ell_j}(\varepsilon) \| \) for all \( j = 1, \dots, n-1 \).

 Note that the multi-indices $\bm{\ell} = (\ell_1, \dots, \ell_n)$ in the sum~\eqref{eq:firstsum} satisfy 
\begin{equation}
|\bm{\ell}| = \ell_1 + \cdots + \ell_n = s.
\end{equation}
Since $K_0(\varepsilon) = 0$ by definition~\eqref{eq:K} (because $K(\varepsilon,\xi)$ vanishes at $\varepsilon = 0$), only terms with $\ell_j \geq 1$ contribute.  
Moreover, as $n \geq 2$, we have $\ell_j \leq s - 1$ for all $j$, so the inductive hypothesis applies,
\begin{equation}
\frac{\pi}{\sqrt{3}\,\eta}  \| B_{\ell_j}(\varepsilon) \| \leq \alpha^{\ell_j - 1} b^{\ell_j} d_{\ell_j}.
\end{equation}
We have
\begin{equation}
\frac{\pi}{\sqrt{3}\,\eta} \| \tilde{\mathcal{K}}_{\ell_1}^{\varepsilon} \cdots \tilde{\mathcal{K}}_{\ell_n}^{\varepsilon}(H)\|
\leq 2^n \alpha^{s - n} b^s d_{\ell_1} \cdots d_{\ell_n} 
= \left( \frac{2}{\alpha} \right)^n \alpha^s b^s d_{\ell_1} \cdots d_{\ell_n}.
\label{eq:adjoint_action_inequality}
\end{equation}
A similar argument yields
\begin{equation}
\frac{\pi}{\sqrt{3}\,\eta} \| \tilde{\mathcal{K}}_{\ell_1}^{\varepsilon} \cdots \tilde{\mathcal{K}}_{\ell_n}^{\varepsilon}(V(\varepsilon)) \|
\leq \left( \frac{2}{\alpha} \right)^n \alpha^{s-1} b^s d_{\ell_1} \cdots d_{\ell_n}.
\end{equation}
We now combine both estimates,
\begin{eqnarray}
    \frac{\pi}{\sqrt{3}\,\eta} \| B_s(\varepsilon) \| 
    \leq \Bigg[&
        \alpha\sum_{n=2}^{s} \frac{1}{n!} \left( \frac{2}{\alpha} \right)^n 
        \sum_{|\bm{\ell}|=s} d_{\ell_1} \cdots d_{\ell_n}\nonumber \\
    & {}+ 
        \sum_{n=1}^{s-1} \frac{1}{n!} \left( \frac{2}{\alpha} \right)^n 
        \sum_{|\bm{\ell}|=s-1} d_{\ell_1} \cdots d_{\ell_n}
    \Bigg]\,\alpha^{s-1}b^s.
\end{eqnarray}
Using the inequality
\begin{equation}
   \sum_{|\bm{\ell}|=s} d_{\ell_1} \cdots d_{\ell_n} \leq d_s,
   \label{eq:Catalan}
\end{equation}
proved in the~\ref{appendix} in the Theorem~\ref{th:property1},
we get
\begin{equation}
    \frac{\pi}{\sqrt{3}\,\eta} \| B_s(\varepsilon) \| 
    \leq\left[ \alpha \sum_{n=2}^s \left( \frac{2}{\alpha} \right)^n \frac{d_s}{n!} + \sum_{n=1}^{s-1} \left( \frac{2}{\alpha} \right)^n \frac{d_{s-1}}{n!} \right] \alpha^{s-1} b^s.
\end{equation}
Since the terms in the sums over $n$ are positive, we can extend the sums to $+\infty$ to get
\begin{equation}
\frac{\pi}{\sqrt{3}\,\eta} \| B_s(\varepsilon) \|
\leq \left[ \alpha \left( \rme^{\frac{2}{\alpha}} - 1 - \frac{2}{\alpha} \right) d_s + ( \rme^{\frac{2}{\alpha}} - 1 ) d_{s-1} \right] \alpha^{s-1} b^s.
\label{eq:inpol}
\end{equation}
Using the Taylor expansion
\begin{equation}
\rme^{\frac{2}{\alpha}} = 1 + \frac{2}{\alpha} + \sum_{n=2}^\infty \frac{1}{n!} \left( \frac{2}{\alpha} \right)^n,
\end{equation}
and observing that $d_{s-1} \leq d_s$, we conclude
\begin{equation}
\frac{\pi}{\sqrt{3}\,\eta} \| B_s(\varepsilon) \|
\leq \left[ (\alpha + 1) ( \rme^{\frac{2}{\alpha}} - 1 ) - 2 \right] d_s \alpha^{s-1}b^s.
\end{equation}
Since $\alpha$ solves~\eqref{eq:alpha}, the quantity in the brackets is equal to $1$, and we obtain the desired bound.
\end{proof}
By using the estimates for the operators $B_s(\varepsilon)$ established in Lemma~\ref{l:2}, we can prove that the operators $\hat{V}(\varepsilon) = \hat{V}(\varepsilon,\varepsilon)$ and $K(\varepsilon)= K(\varepsilon,\varepsilon)$ are bounded for sufficiently small values of $\varepsilon$.  As a preliminary step, we show that the coefficients appearing in the expansions of these operators can be themselves bounded in terms of the norms of the $B_s(\varepsilon)$.
\begin{lemma}\label{l:bounds_coeffs}
For all $s \geq 1$, the operators $\hat{V}_{s-1}(\varepsilon)$ and $K_s(\varepsilon)$ given in Lemma~\ref{l:63} satisfy the bounds
\begin{eqnarray}
    \| \hat{V}_{s-1}(\varepsilon) \| \leq \| B_s(\varepsilon) \|, \label{eq:boundsV}\\
    \| K_s(\varepsilon) \| \leq \frac{\pi}{\sqrt{3}\, \eta} \| B_s(\varepsilon) \|. \label{eq:bounds}
\end{eqnarray}
\end{lemma}

\begin{proof}
According to the Lemma~\ref{l:63}$, \hat{V}_{s-1}(\varepsilon)=[B_s(\varepsilon)]$.
Hence,
\begin{equation}
    \| \hat{V}_{s-1}(\varepsilon) \| \leq \| B_s(\varepsilon) \|.
\end{equation}

The bound~\eqref{eq:bounds} for $\| K_s(\varepsilon) \|$ follows directly from the bound~\eqref{eq:relative} of Lemma~\ref{l:61}, since $K_s(\varepsilon)$ is the solution of the homological equation.
\end{proof}
Now, we are ready to prove the boundedness of the operators $K(\varepsilon)$ and $\hat{V}(\varepsilon)$.

\begin{lemma}\label{th:estimate1}
Let $K(\varepsilon)$ and $\hat{V}(\varepsilon)$ be the diagonal evaluation of the operators $K(\varepsilon,\xi)$ and $\hat{V}(\varepsilon,\xi)$ defined by the expansions~\eqref{eq:K} and~\eqref{eq:V}, namely,
\begin{equation}
    K(\varepsilon) = \sum_{s \geq 1} \varepsilon^s K_s(\varepsilon),  \qquad
     \hat{V}(\varepsilon) = \sum_{s \geq 0} \varepsilon^s \hat{V}_s(\varepsilon), 
     \label{eq:diagseries}
 \end{equation}
with $K_s(\varepsilon)$ and $\hat{V}_s(\varepsilon)$ being the bounded self-adjoint operators given in Lemma~\ref{l:63}\@.
Let 
\begin{equation}
v= \sup_{\varepsilon\in[0,1]}\|V(\varepsilon)\| < +\infty,
\end{equation}
and
\begin{equation}
b=\frac{\pi v}{\sqrt{3}\,\eta}, \qquad \rho=\frac{4\pi\alpha}{\sqrt{3}},
\end{equation}
where $\eta>0$ is the minimal spectral gap of the Hamiltonian $H$, and $\alpha$ is the constant in Table~\ref{tab:constants}\@.

Then, for all $\varepsilon$ such that
\begin{equation}
0\le\varepsilon\le\frac{\eta}{v\rho},
\end{equation}
the series~\eqref{eq:diagseries} converge in norm, and the following bounds hold,
\begin{equation}
\|K(\varepsilon)\|\le \varepsilon b\mathcal D(\alpha\varepsilon b),
\qquad 
\|\hat V(\varepsilon)\|\le v\mathcal{D}(\alpha \varepsilon b),
\end{equation}
where 
\begin{equation}
\mathcal D(y)=\frac{1-\sqrt{1-4y}}{2y}
\label{eq:genCat}
\end{equation}
is the generating function of the Catalan numbers. 
\end{lemma}

\begin{proof}
We start by estimating the norm of the operator $K(\varepsilon)$. Using the bound~\eqref{eq:bounds} and Lemma~\ref{l:2}, we obtain
\begin{equation}
    \| K(\varepsilon,\xi) \| \leq \sum_{\ell \geq 1}  \varepsilon ^\ell \| K_\ell(\xi) \|
    \leq  \varepsilon  b \sum_{\ell \geq 1} d_\ell (\alpha  \varepsilon  b)^{\ell - 1},
\end{equation}
uniformly in $\xi\in[0,1]$.
Moreover, it can be shown that the series of the generating function of the Catalan numbers converges for all $|y| \leq \frac{1}{4}$, namely,
\begin{equation}
    \sum_{\ell \geq 1} d_\ell y^{\ell - 1} = \mathcal{D}(y).
\end{equation}
See  Theorem~\ref{th:Catalan_gen} in~\ref{appendix}\@.
    Hence,
\begin{equation}
    \| K(\varepsilon) \| = \| K(\varepsilon,\varepsilon) \| \leq \sup_{\xi\in[0,1]} \| K(\varepsilon,\xi) \| \leq  \varepsilon  b  \mathcal{D}(\alpha  \varepsilon  b),
\end{equation}
provided that
\begin{equation}
    0\leq \varepsilon  \leq \frac{1}{4 \alpha b}
    = \frac{\sqrt{3}\, \eta}{4 v \pi \alpha}
    = \frac{\eta}{v\rho}.
\end{equation}
The bound for $\hat{V}(\varepsilon)$ follows similarly. Using the bound~\eqref{eq:boundsV}, we can write
\begin{equation}
    \| \hat{V}(\varepsilon,\xi) \| 
    \leq \sum_{\ell \geq 0}  \varepsilon ^\ell \| \hat{V}_\ell(\xi) \| 
    \leq \sum_{\ell \geq 0}  \varepsilon ^\ell \| B_{\ell+1}(\xi) \|.
\end{equation}
And, by using the bound of Lemma~\ref{l:2},
\begin{equation}
  \| B_{\ell+1}(\xi) \| \leq \frac{\sqrt{3}\,\eta}{\pi} \alpha^{\ell} b^{\ell+1} d_{\ell+1} = v (\alpha b)^{\ell} d_{\ell+1},  
\end{equation}
 we get, after shifting the index of the sum,
\begin{equation}
   \fl\qquad \quad \| \hat{V}(\varepsilon) \| = \| \hat{V}(\varepsilon,\varepsilon) \| \leq \sup_{\xi\in[0,1]}\| \hat{V}(\varepsilon,\xi) \| 
    \leq v \sum_{\ell \geq 1} d_\ell (\alpha  \varepsilon  b)^{\ell - 1}
    = v  \mathcal{D}(\alpha  \varepsilon  b),
\end{equation}
again for $0\leq \varepsilon  \leq \eta / (v\rho)$.
\end{proof}
This lemma concludes the construction of the bounded self-adjoint operators $K(\varepsilon)$ and $\hat{V}(\varepsilon)$. 
\subsection{\texorpdfstring{Proof of $(iii)$ and $(vi)$ of Theorem~\ref{thm:fund}}{Proof of (iii) and (vi) of Theorem~\ref{thm:fund}}}
In the previous subsections, we have explicitly constructed the operators $W(\varepsilon) = \rme^{\rmi K(\varepsilon)}$ and $\hat{V}(\varepsilon)$ appearing in Theorem~\ref{thm:fund}\@. We now verify that they satisfy the required properties. Properties $(i)$, $(ii)$, $(iv)$,  and $(v)$ follow directly from the construction. It remains to prove properties $(iii)$ and $(vi)$. We begin with property $(iii)$, namely that, for all $\psi \in D(H)$, we have $\| HW(\varepsilon)\psi \| <+ \infty$. In fact, we can prove a more quantitative bound in terms of the Catalan generating function $\mathcal{D}$.
\begin{lemma}
Let $W(\varepsilon) = \rme^{\rmi K(\varepsilon)}$ with $K(\varepsilon)$ as in~\eqref{eq:diagseries} of Lemma~\ref{th:estimate1},  and assume that 
$0\leq \varepsilon \le\eta/(v\rho)$, with $\rho$, $v$, and $\eta$ as in Theorem~\ref{thm:fund}\@. Then, for all $\psi \in D(H)$,
\begin{equation}
\| HW(\varepsilon)\psi \|
\leq
\rme^{ \varepsilon  b\mathcal{D}(\alpha  \varepsilon  b)}
\,\Bigl(
\| H\psi \|
+
2 v \mathcal{D}(\alpha  \varepsilon  b)
\|\psi\|
\Bigr),
\end{equation}
where $\mathcal{D}$ is the generating function of the Catalan numbers~\eqref{eq:genCat}, $\alpha$ is the constant in Table~\ref{tab:constants}, and 
\begin{equation}\label{eq:b}
 b=\frac{\pi v}{\sqrt{3}\,\eta}.
 \end{equation}
\end{lemma}
\begin{proof}
Recall the expansion of $ K(\varepsilon)$ in~\eqref{eq:diagseries},
where the  $K_s(\varepsilon)$ are solutions the homological equations. By  inequality~\eqref{eq:relative} of Lemma~\ref{l:61}, they satisfy the bound
\begin{equation} \label{eq:relative1}
    \| HK_s (\varepsilon)\psi \| \leq \left( 2\| \psi \| + \frac{\pi}{\sqrt{3}\,\eta} \| H\psi \| \right) \| B_s(\varepsilon) \|, \quad \forall \psi \in D(H).
\end{equation}
This yields
\begin{equation}
    \| HK(\varepsilon) \psi \| \leq \left( 2 \| \psi \| + \frac{\pi}{\sqrt{3}\,\eta} \| H\psi \| \right) \sum_{s \geq 1}  \varepsilon ^s \| B_s(\varepsilon) \|.
\end{equation}
Using the bound of Lemma~\ref{l:2},
\begin{equation}
    \| B_s(\varepsilon) \| \leq \frac{\sqrt{3}\,\eta}{\pi} \alpha^{s-1} d_s b^s,
\end{equation}
we find, for $\varepsilon<\eta/(v\rho)$, that
\begin{equation} \label{eq:boundKeps}
    \| HK(\varepsilon)\psi \| \leq \left( \| H\psi \| + 2 \frac{\sqrt{3}\,\eta}{\pi} \| \psi \| \right)  \varepsilon  b  \mathcal{D}(\alpha  \varepsilon  b),
\end{equation}
where $\mathcal{D}(y)$ is the generating function of the Catalan numbers. 

This shows that $K(\varepsilon) D(H) \subseteq D(H)$, and we now prove by induction that, for all $n \geq 1$,
\begin{equation}
    \| H K(\varepsilon)^n \psi \| \leq \left( \| H\psi \| + 2n \frac{\sqrt{3}\,\eta}{\pi} \| \psi \| \right) (   \varepsilon  b  \mathcal{D}(\alpha  \varepsilon  b) )^n.
\end{equation}
We have already shown the case $n = 1$. Suppose the estimate holds for $n = s-1$,
\begin{equation} \label{eq:induction}
    \| H K(\varepsilon)^{s-1} \psi \| \leq \left( \| H\psi \| + 2(s-1) \frac{\sqrt{3}\,\eta}{\pi} \| \psi \| \right) (   \varepsilon  b  \mathcal{D}(\alpha  \varepsilon  b) )^{s-1}.
\end{equation}
Then, since $K(\varepsilon)\psi \in D(H)$, we can apply~\eqref{eq:induction} to $\varphi = K(\varepsilon)\psi$, obtaining
\begin{eqnarray}
    \| H K(\varepsilon)^s \psi \| 
    &= \| H K(\varepsilon)^{s-1} K(\varepsilon) \psi \| \nonumber\\
    &\leq \left( \| HK(\varepsilon)\psi \| + 2(s-1) \frac{\sqrt{3}\,\eta}{\pi} \| K(\varepsilon)\psi \| \right)  \varepsilon  b  \mathcal{D}(\alpha  \varepsilon  b) \nonumber \\
    &\leq \left( \| H\psi \| + 2s \frac{\sqrt{3}\,\eta}{\pi} \| \psi \| \right) (   \varepsilon  b  \mathcal{D}(\alpha  \varepsilon  b) )^s,
\end{eqnarray}
where we used Theorem~\ref{th:estimate1} and~\eqref{eq:boundKeps}.

We are now ready to estimate $\| HW(\varepsilon) \psi \|$ for all $\psi \in D(H)$,
\begin{eqnarray}
    \| HW(\varepsilon) \psi \| 
    &= \| H \psi + H(W(\varepsilon) - \mathbb{I}) \psi \| \nonumber\\
    &\leq \| H \psi \| + \left\| H \sum_{n \geq 1} \frac{(-\rmi)^n K(\varepsilon)^n}{n!} \psi \right\| \nonumber\\
    &\leq \| H \psi \| + \sum_{n \geq 1} \frac{ \| H K(\varepsilon)^n \psi \| }{n!} \nonumber\\
    &\leq \| H \psi \| + \sum_{n \geq 1} \left( \| H\psi \| + 2n \frac{\sqrt{3}\,\eta}{\pi} \| \psi \| \right) \frac{ (  \varepsilon  b  \mathcal{D}(\alpha  \varepsilon  b) )^n }{n!} \nonumber\\
    &= \rme^{  \varepsilon  b  \mathcal{D}(\alpha  \varepsilon  b)} \left( \| H\psi \| + 2 \frac{\sqrt{3}\,\eta}{\pi}   \varepsilon  b  \mathcal{D}(\alpha  \varepsilon  b) \| \psi \| \right)\nonumber\\
    &=\rme^{  \varepsilon  b  \mathcal{D}(\alpha  \varepsilon  b)}\,\Bigl( \| H\psi \| + 2   \varepsilon  v \mathcal{D}(\alpha  \varepsilon  b) \| \psi \|\Bigr),
\end{eqnarray}
which is valid again for $\varepsilon\leq\eta/(v\rho)$. In the last step, we have used the definition~\eqref{eq:b} of $b$.
\end{proof}
The last point to verify is property~$(vi)$ of Theorem~\ref{thm:fund}\@. 
We first prove the following lemma.

\begin{lemma}\label{lem:Weps_bound}
Let $W(\varepsilon) = \rme^{\rmi K(\varepsilon)}$ with $K(\varepsilon)$ as in~\eqref{eq:diagseries} of Lemma~\ref{th:estimate1},  and assume that 
$0\leq \varepsilon \le\eta/(v\rho)$, with $\rho$, $v$, and $\eta$ as in Theorem~\ref{thm:fund}\@.
Then,
\begin{equation}
    \| W(\varepsilon) - \mathbb{I} \| 
    \le \rme^{ \varepsilon  b \mathcal{D}(\alpha  \varepsilon  b)} - 1,
    \label{eq:inequalityepskam}
\end{equation}
where $\mathcal{D}$ is the generating function of the Catalan numbers~\eqref{eq:genCat}, $\alpha$ is the constant in Table~\ref{tab:constants}, and $b=\pi v/(\sqrt{3}\,\eta)$.
\end{lemma}

\begin{proof}
We estimate
\begin{eqnarray}
    \| W(\varepsilon) - \mathbb{I} \| 
    &= \| \rme^{\rmi K(\varepsilon)} - \mathbb{I}\| 
      = \left\| \sum_{n \ge 1} \frac{(-\rmi K(\varepsilon))^n}{n!} \right\| \nonumber\\
    &\le \sum_{n \ge 1} \frac{\| K(\varepsilon) \|^n}{n!} 
      = \rme^{\| K(\varepsilon) \|} - 1 \le \rme^{ \varepsilon  b \mathcal{D}(\alpha  \varepsilon  b)} - 1,
\end{eqnarray}
for all \(  \varepsilon  \le \eta / (v\rho) \).
The last step follows from Theorem~\ref{th:estimate1}\@.
\end{proof}

Let us define
\begin{equation}
    f_\alpha(x) = \rme^{\frac{x}{4\alpha} \mathcal{D}(\frac{x}{4})} - 1.
\end{equation}
A direct analysis shows that, for \(0 \leq  x \leq 1 \),
\begin{equation}
    f_\alpha(x) \leq ( \rme^{\frac{1}{2\alpha}} - 1 ) x.
    \label{eq:lin_bou}
\end{equation}
Setting \( x =  4\alpha\varepsilon  b \), we obtain
\begin{equation}\label{linear_bound}
    \| W(\varepsilon) - \mathbb{I} \| 
    \le \frac{\beta}{4} \frac{v}{\eta}  \varepsilon,
    \quad \mathrm{for}\quad  \varepsilon  \le \frac{\eta}{v\rho},
\end{equation}
where
\begin{equation}
    \beta = \frac{16\pi\alpha}{\sqrt{3}} 
    ( \rme^{\frac{1}{2\alpha}} - 1).
\end{equation}
This completes the proof of property~$(vi)$.
\begin{remark}
    Alternatively, \( f_\alpha(x) \) may be approximated by a quadratic upper bound,
\begin{equation}\label{quadratic_bound}
    f_\alpha(x) \leq \frac{x}{4\alpha} + c x^2 ,
\end{equation}
for $0\leq x\leq 1$, where
\begin{equation}
    c =   \rme^{\frac{1}{2\alpha}} - 1  - \frac{1}{4\alpha} .
\end{equation}
A comparison between \( f_\alpha(x) \), its optimal linear bound~\eqref{linear_bound},
and the quadratic upper bound~\eqref{quadratic_bound} is shown in Fig.~\ref{fig:enter-label}\@.
\end{remark}
\begin{figure}[ht]
    \centering
    \includegraphics[width=0.7\linewidth]{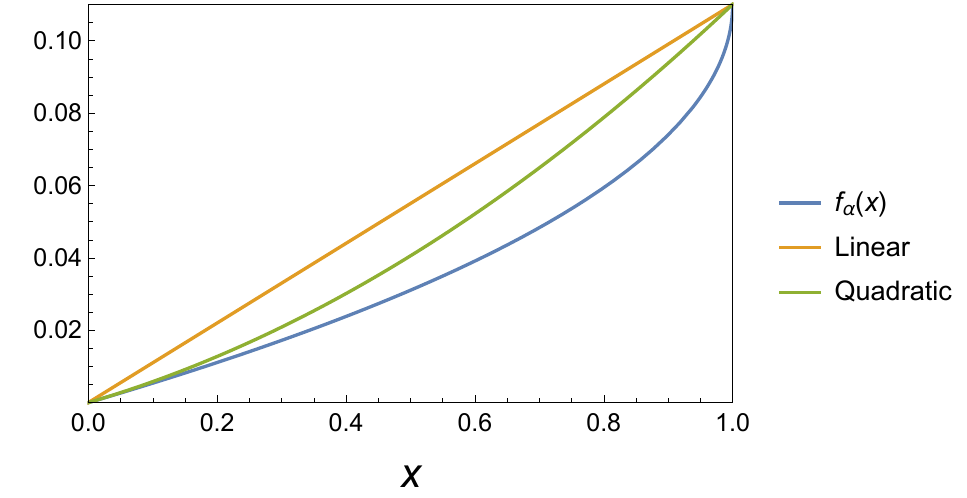}
    \caption{Comparison of $ f_\alpha(x) = \rme^{\frac{x}{4\alpha} \mathcal{D}(\frac{x}{4})} - 1 = \rme^{\frac{1-\sqrt{1-x}}{2\alpha}} - 1$ with its linear bound~\eqref{linear_bound} and quadratic upper bound~\eqref{quadratic_bound}, as a function of $x= 4\alpha \varepsilon b$, with $b=\pi v/(\sqrt{3}\,\eta)$.}
    \label{fig:enter-label}
\end{figure}

The following example illustrates the application of our results to a concrete physical model arising in superconducting circuits.

\begin{example}\label{ex:josephson}
We consider the Hamiltonian of a Josephson junction in an inductive loop~\cite{quasicharge},
\begin{equation}
H_{\mathrm{JJ}}=
4E_C\left(\frac{Q}{2e}\right)^2
+\frac12 E_L\phi^2
-E_J\cos(\phi-\phi_{\mathrm{ext}}).
\end{equation}
Here, $Q$ denotes the charge operator,  and $\phi$ denotes the phase operator describing the phase difference across the junction, while $e$ is the elementary electric charge, $E_C$ is the charging energy, $E_L$ is the inductive energy, $E_J$ is the Josephson energy, and $\phi_{\mathrm{ext}}$ is a constant phase describing the flux biasing of the circuit.

In the phase representation, $\phi$ acts as a multiplication operator on the wave function $\psi(\phi)$, while
\begin{equation}
   \frac{Q}{2e}  = -\rmi \frac{\dd}{\dd \phi} 
\end{equation}
is the corresponding momentum operator.

In order to apply the KAM iteration, we write the Hamiltonian in the form
\begin{equation}
H(\varepsilon)=H+\varepsilon V,
\end{equation}
where $H$ is the unperturbed Hamiltonian and $V$ is the perturbation. The operator $H$ must have pure point spectrum and a nonvanishing minimal spectral gap. The identification of $H$ and $V$ depends on the domain on which the Hamiltonian is defined.

If we work on the whole real line, then
\begin{equation}
D(H)=\{\psi\in H^2(\mathbb{R}) : \phi^2\psi\in L^2(\mathbb{R})\},
\end{equation}
with $H^2(\mathbb{R})$ being the second Sobolev space of square integrable functions with square integrable (weak) second derivative. In this case, we take
\begin{equation}
H=
4E_C\left(\frac{Q}{2e}\right)^2
+\frac12 E_L\phi^2,
\end{equation}
which corresponds to a quantum harmonic oscillator with frequency
\begin{equation}
\omega=\sqrt{8E_CE_L}.
\end{equation}
The operator $H$ satisfies the hypotheses of Theorem~\ref{th:bicomm}, with minimal spectral gap $\eta=\sqrt{8E_CE_L}$. The perturbation is
\begin{equation}
\varepsilon V=-E_J\cos(\phi-\phi_{\mathrm{ext}}),
\end{equation}
which is bounded, with
\begin{equation}
 \varepsilon \|V\|=E_J.
\end{equation}
Therefore, the KAM iteration can be applied whenever
\begin{equation}
 \varepsilon \|V\|\le\frac{\eta}{\rho},
\end{equation}
where $\rho\approx34.7$ is the constant defined in Table~\ref{tab:constants}\@. 
This happens if
\begin{equation}
\frac{E_J/E_C}{\sqrt{E_L/E_C}}\le\frac{2\sqrt{2}}{\rho}.
\end{equation}

If, on the other hand, we work on the circle, the Hilbert space is $L^2(\mathbb{S})\simeq L^2([0,2\pi])$, 
and the Hamiltonian is defined on the second Sobolev space on the circle,
\begin{equation}
D(H)= H^2(\mathbb{S}),
\end{equation}
In this case, we choose
\begin{equation}
H=
4E_C\left(\frac{Q}{2e}\right)^2,
\end{equation}
whose spectrum is discrete with minimal spectral gap
\begin{equation}
\eta=4E_C.
\end{equation}
The perturbation is then
\begin{equation}
\varepsilon V=
\frac12 E_L\phi^2
-E_J\cos(\phi-\phi_{\mathrm{ext}}).
\end{equation}
Since $\phi$ is bounded ($\|\phi\|=2\pi$), we have
\begin{equation}
 \varepsilon \|V\|\le2\pi^2E_L+E_J.
\end{equation}
In this case, the KAM iteration can be applied provided that
\begin{equation}
\frac{\pi^2}{2}\frac{E_L}{E_C}+\frac{E_J}{4E_C}<\frac{1}{\rho}.
\end{equation}
\end{example}
\section*{Conclusions}

In this work, we analyzed the wandering range of symmetries that are robust against perturbations of the Hamiltonian. In particular, for Hamiltonians with pure point spectrum, we investigated its dependence on the strength of the perturbation and identified sufficient conditions for which it scales linearly with it.  We showed that, for states in a dense subspace given by the linear span of eigenvectors of the unperturbed Hamiltonian and for finite-rank symmetries, the wandering range is of order $\varepsilon$. These results rely on Kato's perturbation theory together with the results of Ref.~\cite{rob}.

We then studied the wandering range of completely robust symmetries under bounded perturbations. In this case as well, it scales as $\varepsilon$; however, a stronger statement can be made: it admits an explicit \emph{norm} bound in terms of the perturbation strength and the minimal spectral gap of the unperturbed Hamiltonian.
This bound generalizes the corresponding result for finite-dimensional systems to unbounded Hamiltonians, and sharpens it by removing the dependence on the number of distinct eigenvalues.

The derivation of this bound constitutes the main technical contribution of the second part of the paper. To obtain it, we employed a Schrieffer-Wolff transformation, constructed through a KAM iteration scheme, analogous to the one appearing in classical mechanics. In order to control the various steps of the iteration, we used a property of the Catalan numbers, which naturally arise in the combinatorial structure of the expansion.

To ensure the convergence of the KAM iteration, we considered Hamiltonians with pure point spectrum and a nonvanishing minimal spectral gap. These conditions are satisfied by many  physical models, such as (any relatively bounded perturbation of) the harmonic oscillator and, in general, any quantum particle in a confining potential. However, there are other physically relevant situations in which they fail. In solid-state physics, one typically encounters spectra consisting of continuous energy bands separated by gaps. On the other hand, in molecular systems, one often deals with Hamiltonians with pure point spectrum but vanishing minimal spectral gap, as in the paradigmatic case of the hydrogen atom. Extending the present techniques to these situations would be an interesting direction for future work.

\ack
KY was supported by JSPS KAKENHI Grant No.~JP24K06904 from the Japan Society for the Promotion of Science (JSPS)\@.
PF, ML and VV acknowledge support from the Italian National Group of Mathematical Physics (GNFM-INdAM) and from PNRR MUR projects CN00000013 -``Italian National Centre on HPC, Big Data and Quantum Computing''. 
PF and VV acknowledge support from INFN through the project ``QUANTUM'' and from the Italian funding within the ``Budget MUR - Dipartimenti di Eccellenza 2023--2027''  - Quantum Sensing and Modelling for One-Health (QuaSiModO).
DB and PF acknowledge support from PNRR MUR project PE0000023-NQSTI.

\appendix
\renewcommand{\thedefinition}{A.\arabic{definition}}
\renewcommand{\thetheorem}{A.\arabic{theorem}}
\renewcommand{\theexample}{A.\arabic{example}}
\section{Catalan numbers}\label{appendix}
In this appendix, we give a brief review of the Catalan numbers, the combinatorial sequence crucial in the KAM iteration.

The Catalan numbers are defined by
\begin{equation}
    d_1 = 1, \qquad
    d_s = \sum_{\ell=1}^{s-1} d_\ell d_{s-\ell}.
\end{equation}

First, we present an important inequality this sequence satisfies.

\begin{theorem}\label{th:property1}
Let $\bm{\ell} = (\ell_1, \dots, \ell_n)\in\mathbb{N}^n$ be a multi-index with \(\ell_j \geq 1\). Then,
\begin{equation}
    \sum_{|\bm{\ell}|=s} d_{\ell_1} \cdots d_{\ell_n} \leq d_s,
    \label{eq:property1}
\end{equation}
where \(|\bm{\ell}| = \ell_1 + \cdots + \ell_n\).
\end{theorem}

\begin{proof}
We prove~\eqref{eq:property1} by induction on \(n\).
For \(n=1\), the claim trivially holds,
\begin{equation}
    \sum_{\ell=s} d_\ell = d_s.
\end{equation}
Assume that the property holds for \(n = m\), and let us prove it for \(n = m+1\). Let $\bm{\ell}=(\ell_1,\dots,\ell_{m+1})$, with $\ell_j\ge1$. We get
\begin{eqnarray}
    \sum_{|\bm{\ell}|=s} d_{\ell_1} \cdots d_{\ell_{m+1}} 
    &=& \sum_{\ell_{m+1}=1}^{s - m} d_{\ell_{m+1}} \sum_{|\bm{\ell}|=s - \ell_{m+1}} d_{\ell_1} \cdots d_{\ell_m} \nonumber\\
    &\leq& \sum_{\ell_{m+1}=1}^{s - m} d_{\ell_{m+1}} d_{s - \ell_{m+1}} 
    \leq \sum_{\ell_{m+1}=1}^{s-1} d_{\ell_{m+1}} d_{s - \ell_{m+1}} = d_s.
\end{eqnarray}
\end{proof}
In the next subsection, we introduce and calculate the generating function of the Catalan numbers.
\subsection{Generating function}

The generating function of the Catalan numbers is defined by the series
\begin{equation}
    \mathcal{D}(x) = \sum_{s \geq 1} d_s x^{s-1}.
    \label{eq:Generating}
\end{equation}

It is possible to find an analytical expression for \(\mathcal{D}(x)\), as stated in the following theorem.

\begin{theorem}\label{th:Catalan_gen}
Let \(\mathcal{D}(x)\) be the generating function of the Catalan numbers defined in~\eqref{eq:Generating}. Then,
\begin{equation}
    \mathcal{D}(x) = \frac{1 - \sqrt{1 - 4x}}{2x}.
    \label{eq:genfunCat}
\end{equation}
\end{theorem}

\begin{proof}
Consider the square of \(\mathcal{D}(x)\),
\begin{equation}
    \mathcal{D}^2(x) = \left(\sum_{s_1 \geq 1} d_{s_1} x^{s_1 - 1}\right)\left(\sum_{s_2 \geq 1} d_{s_2} x^{s_2 - 1}\right) = \sum_{s_1, s_2 \geq 1} d_{s_1} d_{s_2} x^{(s_1 + s_2) - 2}.
\end{equation}

Introduce the sum variable \( s = s_1 + s_2 \). Then,
\begin{equation}
    \mathcal{D}^2(x) = \sum_{s \geq 2} x^{s-2} \sum_{s_1=1}^{s-1} d_{s_1} d_{s - s_1} = \frac{1}{x} \sum_{s \geq 2} d_s x^{s-1} = \frac{\mathcal{D}(x) - 1}{x},
\end{equation}
where in the last step we used the definition~\eqref{eq:Generating} and \( d_1 = 1 \).

Thus, \(\mathcal{D}(x)\) satisfies the quadratic equation
\begin{equation}
    x \mathcal{D}^2(x) - \mathcal{D}(x) + 1 = 0.
\end{equation}

Solving for \(\mathcal{D}(x)\) and choosing the branch satisfying \(\mathcal{D}(0) = 1\), we obtain~\eqref{eq:genfunCat}.
\end{proof}

\newcommand{\newblock}{}
\bibliographystyle{iop-title-hyperref}
\bibliography{name}

@Article{symmetry,
  author    = {Gross, David J.},
  journal   = {Proc. Natl. Acad. Sci. U.S.A.},
  title     = {The role of symmetry in fundamental physics},
  year      = {1996},
  issn      = {1091-6490},
  month     = dec,
  number    = {25},
  pages     = {14256--14259},
  volume    = {93},
  doi       = {10.1073/pnas.93.25.14256},
  publisher = {Proceedings of the National Academy of Sciences},
  url       = {https://doi.org/10.1073/pnas.93.25.14256},
}

@Book{Wigner_symmetry,
  author    = {Wigner, Eugene P.},
  publisher = {MIT Press},
  title     = {Symmetries and Reflections: Scientific Essays},
  year      = {1970},
  address   = {Cambridge, MA},
}

@Book{Sakurai,
  author    = {Sakurai, J. J. and Napolitano, Jim},
  publisher = {Cambridge University Press},
  title     = {Modern Quantum Mechanics},
  year      = {2020},
  address   = {Cambridge},
  edition   = {3rd},
  isbn      = {9781108473224},
  month     = sep,
  doi       = {10.1017/9781108587280},
  url       = {https://doi.org/10.1017/9781108587280},
}

@Article{rob,
  author    = {Facchi, Paolo and Ligab{\`{o}}, Marilena and Viesti, Vito},
  journal   = {J. Phys. A: Math. Theor.},
  title     = {Robustness of quantum symmetries against perturbations},
  year      = {2025},
  issn      = {1751-8121},
  month     = mar,
  number    = {12},
  pages     = {125305},
  volume    = {58},
  doi       = {10.1088/1751-8121/adbfe5},
  publisher = {IOP Publishing},
  url       = {https://doi.org/10.1088/1751-8121/adbfe5},
}

@Article{harm,
  author        = {Gr{\'e}bert, B. and Thomann, L.},
  journal       = {Commun. Math. Phys.},
  title         = {KAM for the Quantum Harmonic Oscillator},
  year          = {2011},
  issn          = {1432-0916},
  month         = aug,
  number        = {2},
  pages         = {383--427},
  volume        = {307},
  doi           = {10.1007/s00220-011-1327-5},
  publisher     = {Springer Science and Business Media LLC},
  url           = {https://doi.org/10.1007/s00220-011-1327-5},
}

@InBook{bhatia1989offdiagonal,
  author    = {Bhatia, Rajendra and Choi, Man-Duen and Davis, Chandler},
  editor    = {Dym, H. and Goldberg, S. and Kaashoek, M. A. and Lancaster, P.},
  pages     = {151--164},
  publisher = {Birkh{\"{a}}user},
  title     = {Comparing a Matrix to its Off-Diagonal Part},
  year      = {1989},
  address   = {Basel},
  isbn      = {9783034891448},
  series    = {Operator Theory: Advances and Applications},
  volume    = {40--41},
  booktitle = {The Gohberg Anniversary Collection},
  doi       = {10.1007/978-3-0348-9144-8_4},
  url       = {https://doi.org/10.1007/978-3-0348-9144-8_4},
}

@Book{stanley2015catalan,
  author    = {Stanley, Richard P.},
  publisher = {Cambridge University Press},
  title     = {Catalan Numbers},
  year      = {2015},
  address   = {Cambridge},
  isbn      = {9781107075091},
  month     = mar,
  doi       = {10.1017/CBO9781139871495},
  url       = {https://doi.org/10.1017/CBO9781139871495},
}

@Article{kam,
  author    = {Burgarth, Daniel and Facchi, Paolo and Nakazato, Hiromichi and Pascazio, Saverio and Yuasa, Kazuya},
  journal   = {Phys. Rev. Lett.},
  title     = {Kolmogorov-{Arnold}-{Moser} Stability for Conserved Quantities in Finite-Dimensional Quantum Systems},
  year      = {2021},
  issn      = {1079-7114},
  month     = apr,
  number    = {15},
  pages     = {150401},
  volume    = {126},
  doi       = {10.1103/PhysRevLett.126.150401},
  issue     = {15},
  numpages  = {6},
  publisher = {American Physical Society},
  url       = {https://doi.org/10.1103/PhysRevLett.126.150401},
}

@Article{quasicharge,
  author        = {Pechenezhskiy, Ivan V. and Mencia, Raymond A. and Nguyen, Long B. and Lin, Yen-Hsiang and Manucharyan, Vladimir E.},
  journal       = {Nature},
  title         = {The superconducting quasicharge qubit},
  year          = {2020},
  issn          = {1476-4687},
  month         = sep,
  number        = {7825},
  pages         = {368--371},
  volume        = {585},
  date          = {2020/09/01},
  doi           = {10.1038/s41586-020-2687-9},
  id            = {Pechenezhskiy2020},
  isbn          = {1476-4687},
  publisher     = {Springer Science and Business Media LLC},
  url           = {https://doi.org/10.1038/s41586-020-2687-9},
}

@Article{simul,
  author        = {Sarovar, Mohan and Zhang, Jun and Zeng, Lishan},
  journal       = {Eur. Phys. J. Quantum Technol.},
  title         = {{Reliability of analog quantum simulation}},
  year          = {2017},
  issn          = {2196-0763},
  month         = jan,
  number        = {1},
  pages         = {1},
  volume        = {4},
  doi           = {10.1140/epjqt/s40507-016-0054-4},
  publisher     = {Springer Science and Business Media LLC},
  url           = {https://doi.org/10.1140/epjqt/s40507-016-0054-4},
}

@Article{simul2,
  author    = {Georgescu, I. M. and Ashhab, S. and Nori, Franco},
  journal   = {Rev. Mod. Phys.},
  title     = {Quantum simulation},
  year      = {2014},
  issn      = {1539-0756},
  month     = mar,
  number    = {1},
  pages     = {153--185},
  volume    = {86},
  doi       = {10.1103/RevModPhys.86.153},
  issue     = {1},
  numpages  = {33},
  publisher = {American Physical Society},
  url       = {https://doi.org/10.1103/RevModPhys.86.153},
}

@Article{qsimul,
  author    = {Schwenk, Iris and Reiner, Jan-Michael and Zanker, Sebastian and Tian, Lin and Lepp\"akangas, Juha and Marthaler, Michael},
  journal   = {Phys. Rev. A},
  title     = {Reconstructing the ideal results of a perturbed analog quantum simulator},
  year      = {2018},
  issn      = {2469-9934},
  month     = apr,
  number    = {4},
  pages     = {042310},
  volume    = {97},
  doi       = {10.1103/PhysRevA.97.042310},
  issue     = {4},
  numpages  = {13},
  publisher = {American Physical Society},
  url       = {https://doi.org/10.1103/PhysRevA.97.042310},
}

@Article{gapless,
  author     = {Facchi, Paolo and Ligab{\`{o}}, Marilena},
  journal    = {J. Math. Phys.},
  title      = {Stability of the gapless pure point spectrum of self-adjoint operators},
  year       = {2024},
  issn       = {1089-7658},
  month      = mar,
  number     = {3},
  pages      = {032102},
  volume     = {65},
  doi        = {10.1063/5.0187017},
  publisher  = {AIP Publishing},
  url        = {https://doi.org/10.1063/5.0187017},
}

@Book{Arnold_DynSysIII_1988,
  author    = {Arnold, V. I. and Kozlov, V. V. and Neishtadt, A. I.},
  publisher = {Springer-Verlag},
  title     = {Dynamical Systems III: Mathematical Aspects of Classical and Celestial Mechanics},
  year      = {1988},
  address   = {Berlin, Heidelberg},
  isbn      = {978-3-662-02535-2},
  series    = {Encyclopaedia of Mathematical Sciences},
  volume    = {3},
  doi       = {10.1007/978-3-662-02535-2},
  issn      = {0938-0396},
  url       = {https://doi.org/10.1007/978-3-662-02535-2},
}

@Book{Dumas,
  author    = {Dumas, H. Scott},
  publisher = {World Scientific},
  title     = {The {KAM} Story: A Friendly Introduction to the Content, History, and Significance of Classical {Kolmogorov}-{Arnold}-{Moser} Theory},
  year      = {2014},
  address   = {Singapore},
  isbn      = {9789814556590},
  month     = apr,
  doi       = {10.1142/8955},
  url       = {https://doi.org/10.1142/8955},
}

@Article{schrieffer_wolff,
  author    = {Schrieffer, J. R. and Wolff, P. A.},
  journal   = {Phys. Rev.},
  title     = {Relation between the {Anderson} and {Kondo} Hamiltonians},
  year      = {1966},
  issn      = {0031-899X},
  month     = sep,
  number    = {2},
  pages     = {491--492},
  volume    = {149},
  doi       = {10.1103/PhysRev.149.491},
  publisher = {American Physical Society},
  url       = {https://doi.org/10.1103/PhysRev.149.491},
}

@Article{bravyi_schrieffer_wolff,
  author    = {Bravyi, Sergey and DiVincenzo, David P. and Loss, Daniel},
  journal   = {Ann. Phys.},
  title     = {Schrieffer-{Wolff} transformation for quantum many-body systems},
  year      = {2011},
  issn      = {0003-4916},
  month     = oct,
  number    = {10},
  pages     = {2793--2826},
  volume    = {326},
  doi       = {10.1016/j.aop.2011.06.004},
  publisher = {Elsevier BV},
  url       = {https://doi.org/10.1016/j.aop.2011.06.004},
}

@Article{Cederbaum_1989,
  author    = {Cederbaum, L. S. and Schirmer, J. and Meyer, H. D.},
  journal   = {J. Phys. A: Math. Gen.},
  title     = {Block diagonalisation of {Hermitian} matrices},
  year      = {1989},
  issn      = {1361-6447},
  month     = jul,
  number    = {13},
  pages     = {2427--2439},
  volume    = {22},
  doi       = {10.1088/0305-4470/22/13/035},
  publisher = {IOP Publishing},
  url       = {https://doi.org/10.1088/0305-4470/22/13/035},
}

@Article{Szabo2025robust,
  author    = {Szab{\'{o}}, Zsolt and Gehr, Stefan and Facchi, Paolo and Yuasa, Kazuya and Burgarth, Daniel and Lonigro, Davide},
  journal   = {Phys. Rev. A},
  title     = {Robust quantification of spectral transitions in perturbed quantum systems},
  year      = {2025},
  issn      = {2469-9934},
  month     = sep,
  number    = {3},
  pages     = {032202},
  volume    = {112},
  doi       = {10.1103/j44p-13j7},
  publisher = {American Physical Society (APS)},
  url       = {https://doi.org/10.1103/j44p-13j7},
}

@Book{BR,
  author        = {Bratteli, Ola and Robinson, Derek W.},
  publisher     = {Springer},
  title         = {Operator Algebras and Quantum Statistical Mechanics 1},
  year          = {1987},
  address       = {Berlin},
  edition       = {2nd},
  isbn          = {9783662025208},
  doi           = {10.1007/978-3-662-02520-8},
  url           = {https://doi.org/10.1007/978-3-662-02520-8},
}

@Book{Kato1995,
  author        = {Kato, Tosio},
  publisher     = {Springer},
  title         = {Perturbation Theory for Linear Operators},
  year          = {1995},
  address       = {Berlin},
  edition       = {2nd},
  isbn          = {978-3-540-58661-6},
  doi           = {10.1007/978-3-642-66282-9},
  issn          = {1431-0821},
  url           = {https://doi.org/10.1007/978-3-642-66282-9},
}

@PhdThesis{ViestiPhD,
  author = {Viesti, Vito},
  school = {Universit{\`{a}} degli Studi di Bari Aldo Moro},
  title  = {Quantum Symmetries and the Robustness of Dynamics {}},
  year   = {2026},
  url    = {https://hdl.handle.net/11586/568884},
}

@Article{Eternal,
  author    = {Burgarth, Daniel and Facchi, Paolo and Nakazato, Hiromichi and Pascazio, Saverio and Yuasa, Kazuya},
  journal   = {Phys. Rev. A},
  title     = {Eternal adiabaticity in quantum evolution},
  year      = {2021},
  issn      = {2469-9934},
  month     = mar,
  number    = {3},
  pages     = {032214},
  volume    = {103},
  doi       = {10.1103/PhysRevA.103.032214},
  issue     = {3},
  numpages  = {23},
  publisher = {American Physical Society (APS)},
  url       = {https://doi.org/10.1103/PhysRevA.103.032214},
}

@Book{arnold1989classical,
  author    = {Arnold, V. I.},
  publisher = {Springer-Verlag},
  title     = {Mathematical Methods of Classical Mechanics},
  year      = {1989},
  address   = {New York},
  edition   = {2nd},
  isbn      = {978-0-387-96890-2},
  doi       = {10.1007/978-1-4757-2063-1},
  issn      = {0072-5285},
  journal   = {Graduate Texts in Mathematics},
  url       = {https://doi.org/10.1007/978-1-4757-2063-1},
}

@Book{bhatia1997matrix,
  author    = {Bhatia, Rajendra},
  publisher = {Springer-Verlag},
  title     = {Matrix Analysis},
  year      = {1997},
  address   = {New York},
  isbn      = {978-0-387-94846-1},
  doi       = {10.1007/978-1-4612-0653-8},
  issn      = {0072-5285},
  url       = {https://doi.org/10.1007/978-1-4612-0653-8},
}

@Book{teschl2014quantum,
  author    = {Teschl, Gerald},
  publisher = {American Mathematical Society},
  title     = {Mathematical Methods in Quantum Mechanics: With Applications to Schr{\"o}dinger Operators},
  year      = {2014},
  address   = {Providence, RI},
  edition   = {2nd},
  isbn      = {978-1-4704-1704-8},
}
\end{document}